\documentclass[sigconf]{acmart}
\AtBeginDocument{%
  }

\setcopyright{acmlicensed}
\copyrightyear{2026}
\acmYear{2026}
\acmDOI{XXXXXXX.XXXXXXX}
\acmConference[SIGMOD '26]{Make sure to enter the correct
  conference title from your rights confirmation email}{May 31--June 05,
  2026}{Bengaluru, India}

\acmISBN{978-1-4503-XXXX-X/2018/06}

\newcommand{\ourWorkName}{SLW-Graph}
\newcommand{\graph}{\emph{\ourWorkName}}
\newcommand{\graphs}{\emph{\ourWorkName}s}

\newcommand{\ourSystemName}{Brook-2PL}
\newcommand{\ourSystem}{\emph{\ourSystemName}}

\newcommand{\cat}[1]{\smallskip\noindent\textbf{#1.}}

\newcommand{\re}[1]{#1}
\newcommand{\two}[1]{#1}
\newcommand{\four}[1]{#1}
\newcommand{\five}[1]{#1}

\usepackage{subcaption}
\usepackage{mathtools}
\usepackage{algorithm}
\usepackage{algpseudocode}
\usepackage{enumitem}
\usepackage{xcolor}
\usepackage{listings}
\usepackage{caption}

\usepackage{afterpage} 

\lstdefinelanguage{Transaction}{
    morekeywords=[1]{read, insert, update, delete, Read, Delete, Write, Insert, for, in},
    morekeywords=[2]{Transaction},
    morekeywords=[3]{AddListing, BuyListing, ReadItems},
    keywordstyle=[1]\color{blue},
    keywordstyle=[2]\color{violet},
    keywordstyle=[3]\bfseries\color{violet},
    comment=[l]{//},
    commentstyle=\color{gray}
}

\author{Farzad Habibi}
\affiliation{%
  \institution{University of California, Irvine}
}
\email{habibif@uci.edu}

\author{Juncheng Fang}
\affiliation{%
  \institution{University of California, Irvine}
}
\email{junchf1@uci.edu}

\author{Tania Lorido-Botran}
\affiliation{%
  \institution{Roblox, Northeastern University}
}
\email{tbotran@roblox.com}
\email{t.loridobotran@northeastern.edu}

\author{Faisal Nawab}
\affiliation{%
  \institution{University of California, Irvine}
}
\email{nawabf@uci.edu}

\makeatletter
\def\@ACM@checkaffil{
    \if@ACM@instpresent\else
    \ClassWarningNoLine{\@classname}{No institution present for an affiliation}%
    \fi
    \if@ACM@citypresent\else
    \ClassWarningNoLine{\@classname}{No city present for an affiliation}%
    \fi
    \if@ACM@countrypresent\else
        \ClassWarningNoLine{\@classname}{No country present for an affiliation}%
    \fi
}
\makeatother

\begin{document}
\title{\ourSystemName: Tolerating High Contention Workloads with A Deadlock-Free Two-Phase Locking Protocol}

\begin{abstract}

\sloppy The problem of hotspots remains a critical challenge in high-contention workloads for concurrency control (CC) protocols.
Traditional concurrency control approaches encounter significant difficulties under high contention, resulting in excessive transaction aborts and deadlocks.
In this paper, we propose \ourSystem, a novel two-phase locking (2PL) protocol that (1) introduces \graph\ for deadlock-free transaction execution, and (2) proposes \emph{partial transaction chopping} for early lock release.
Previous methods suffer from transaction aborts that lead to wasted work and can further burden the system due to their cascading effects.
\ourSystem\ addresses this limitation by statically analyzing a new graph-based dependency structure called \graph, enabling deadlock-free two-phase locking through predetermined lock acquisition.
\ourSystem\ also reduces contention by enabling early lock release using partial transaction chopping and static transaction analysis. We overcome the inherent limitations of traditional transaction chopping by providing a more flexible chopping method.
Evaluation using both our synthetic online game store workload and the TPC-C benchmark shows that \ourSystem\ significantly outperforms state-of-the-art CC protocols. 
\ourSystem\ achieves an average speed-up of \re{\(2.86\times\)} while reducing tail latency (p95) by \re{\(48\%\)} in the TPC-C benchmark.


\vspace{-5pt}

\end{abstract}




\settopmatter{printfolios=true} 
\maketitle

\vspace{-5pt}

\section{Introduction}


The rapid growth of cloud computing provides an excellent infrastructure for data management, offering scalable and flexible resources to meet the growing demands of modern applications~\cite{jacobs2007ruminations}. Many cloud data management systems have been proposed and widely adopted~\cite{lakshman2010cassandra, chang2008bigtable, khetrapal2006hbase,baker2011megastore}. However, achieving both scalability and consistency under a high contention workload remains challenging~\cite{appuswamy2017analyzingUnderHighContntion, yu2014ccEvaluation}.

Transaction processing systems use advanced concurrency control protocols to leverage parallelism. However, these protocols still suffer from conflicts that reduce concurrency and throughput~\cite{bernstein1987concurrency, Tanabe2020ccbench, guo2021bamboo, wang2016ics3, bruhwiler2022analyzing, fang2022pelopartition}.
In highly contentious workloads, hotspots---frequently accessed small sets of database records---pose significant challenges~\cite{appuswamy2017analyzingUnderHighContntion, yu2014ccEvaluation}. To maintain strong isolation and serializability, concurrency control protocols must serialize transactions around these hotspots, even though they represent only a small fraction of total transactions. Under traditional concurrency protocols, transactions encountering hotspots may be forced to either wait or abort and restart, further impacting system performance~\cite{guo2021bamboo, Thomasian1998DistributedOC}.


To improve performance, prior production systems~\cite{larson2011high, diaconu2013hekaton, eldeeb2016ms_orleans} and research efforts~\cite{guo2021bamboo, wang2016ics3, ding2018improvingOCC, graefe2013controlledLockViolation, faleiro2017earlyWriteVisibility, kimura2012efficientLocking} have explored early write visibility through various methods. For instance, Bamboo~\cite{guo2021bamboo} modifies pessimistic locking in the 2PL protocol to allow locks to be released earlier, enhancing parallelism. Their solution involves tracking dirty reads and cascadingly aborting transactions if one transaction aborts.
However, this approach continues to experience transaction aborts under high contention, resulting in wasted work and degraded performance. In fact, cascading aborts become especially problematic in heavily contentious workloads, generating additional wasted work and further straining system resources.

Transaction chopping approaches~\cite{shasha1995transactionChopping}, such as IC3~\cite{wang2016ics3}, also enable early write visibility through static analysis of transactions to improve performance. 
However, transaction chopping is rigid in that isolated chops of a transaction follow their own concurrency control protocol and do not interact. This limits the flexibility of acceptable chopping configurations and thus makes them impractical in real applications. These approaches also continue to experience transaction aborts under high contention.

In this work, we propose \ourSystem, a deadlock-free 2PL protocol, to address the problem of hotspots in high-contention workloads by eliminating transaction aborts and enabling earlier lock release.

\ourSystem\ statically analyzes a highly contentious set of transactions using a specific dependency graph called \graph\ (Serializable Wait For Lock Graph). \ourSystem\ eliminates potential deadlocks by enforcing a predefined locking order within transactions and adjusting certain lock acquisitions to occur earlier in the transaction sequence.
Additionally, we show that static analysis of transactions and the removal of cyclic dependencies (deadlocks) can facilitate the early release of locks through a concept we call \emph{partial transaction chopping}. Traditional transaction chopping mechanisms~\cite{zhang2013transactionChain, shasha1995transactionChopping} are rigid and impose strict conditions that limit their flexibility in chopping transactions. We relax these assumptions, enabling more flexible and serializable choppings. The main insight of our mechanism is to allow chopping transactions in a \emph{partial}---rather than exclusive---manner.

In summary, this paper makes the following contributions: 

\begin{itemize}[leftmargin=*, topsep=1pt]
    \item We develop \ourSystem, a new deadlock-free concurrency control protocol that eliminates transaction aborts and significantly reduces wasted work, improving the parallelism of transactions.
    
    \item We propose \graph, a graph-based representation of transactions, and introduce lock manipulation techniques to eliminate potential deadlocks.
    
    \item We combine \graph\ with transaction chopping to enable early lock release, enhancing performance by increasing parallelism. Our approach to write visibility surpasses others by identifying more opportunities for early write visibility based on preprocessing and transaction analysis. Transaction chopping has inherent limitations, which we overcome by enabling a greater number of transactions to be chopped via \emph{partial} transaction chopping.

    \item We design a static analysis algorithm that leverages \graph\ and \emph{partial transaction chopping} to analyze a set of highly contentious transactions, generating optimized and deadlock-free transactions. Additionally, \ourSystem\ generalizes to all workloads by enabling dynamic transaction execution.

    \item We evaluate \ourSystem\ under both a synthetic online game store workload and the TPC-C benchmark~\cite{tpc-c}, comparing it with our baseline consisting of state-of-the-art concurrency control protocols. \ourSystem\ demonstrates an average speed-up of \re{\(2.86\times\)} while reducing latency by \re{\(48\%\)} in the TPC-C benchmark.

\end{itemize}

\vspace{-5pt}
\cat{Paper Organization}
The remainder of this paper is organized as follows: Section~\ref{sec:Background} provides motivation and background information on database transactions, two-phase locking (2PL), and transaction chopping. Section~\ref{sec:Design} introduces the design of \ourSystem, including static analysis of transactions using \graph\ and our proposed \emph{partial transaction chopping} method, as well as our approach to handle dynamic transactions. Section~\ref{sec:Evaluation} presents experimental results evaluating \ourSystem. Section~\ref{sec:Related_work} reviews relevant literature, and Section~\ref{sec:Conclusion} concludes the paper.



\vspace{-1em}
\section{Background and Motivation}
\label{sec:Background}

\subsection{Database Transactions}

A database transaction is a sequence of operations executed as a single unit of work on a database, ensuring the ACID (atomicity, consistency, isolation, and durability) properties. A transaction \( T \) can be defined as an ordered set of operations: \(
T = \{ o_1, o_2, \dots, o_n \}
\), where each \( o_i \) represents an operation (such as read or write) on data items within the database.



Serializability in a database system ensures concurrent execution of transactions by enforcing rules that prevent unexpected or incorrect results. Each transaction is treated as if it is the only transaction running, thereby maintaining data integrity. A concurrent schedule is considered serializable if it is equivalent to some sequential execution of its transactions. A stricter form, conflict serializability, holds when a schedule is conflict equivalent to a serial schedule—meaning it involves the same transactions and maintains the relative order of all conflicting operations.




\cat{Two-Phase Locking (2PL)}
%
Two-phase locking (2PL) is the most widely used class of concurrency control in database systems. In 2PL, reads and writes are synchronized through explicit locks: shared mode for read operations and exclusive mode for write operations.
A transaction can operate on a resource only if it acquires the necessary locks.
To maintain serializability, 2PL enforces two key rules~\cite{bernstein1987concurrency}:
(1) Conflicting locks are not allowed at the same time on the same data. Two locks conflict if they both lock the same resource and at least one of them is in exclusive mode.
(2) A transaction cannot acquire additional locks once it has released any.


%
Two-phase locking enforces a growing-then-shrinking lock acquisition rule to ensure serializability, but this can lead to deadlocks, mostly resolved through detection or prevention strategies.
In deadlock detection, the system maintains a central wait-for graph that tracks dependencies and periodically checks for cycles during runtime. However, maintaining this graph can create a scalability bottleneck~\cite{yu2014ccEvaluation}.
In deadlock prevention, the system allows only a subset of transactions to wait on locks based on specific criteria. Wound-Wait and Wait-Die are two popular deadlock prevention strategies within 2PL~\cite{rosenkrantz1978systemLevelConcurrency}.
In both deadlock handling mechanisms, the database system aborts a transaction to break the deadlock cycle, resulting in wasted CPU resources and failing to prevent the same deadlock scenario from reoccurring.

\subsection{Transaction Chopping}
\label{sec:Background_tc}


Transaction chopping breaks down a transaction \five{type} into smaller subtransactions, enabling the independent execution of each subtransaction. This independence allows locks to be held only during the execution of each piece. As a result, locks are released earlier, reducing contention and improving concurrency. This approach, referred to as piecewise execution, ensures that locks are held only during the execution of each subtransaction.

However, breaking down a transaction can pose challenges for maintaining serializability. To ensure serializability, transaction chopping relies on a graph called the SC-Graph. This graph verifies that the chopping and piecewise execution of subtransactions are valid and maintain a serializable execution order.

\cat{Example}
Consider a transaction \five{type} \( T = W(X)W(Y)R(Z) \) that accesses three different \five{object types (e.g. tables)}: \( X \), \( Y \), and \( Z \). This transaction \five{type} can be chopped into two subtransactions: \( T_{1} = W(X)W(Y) \) and \( T_2 = R(Z) \), allowing each subtransaction to release its locks after its execution is finished. To process the chain of subtransactions, the system will first execute \( T_1 \), followed by \( T_2 \).

  \begin{figure}
    \centering
    \includegraphics[width=\linewidth]{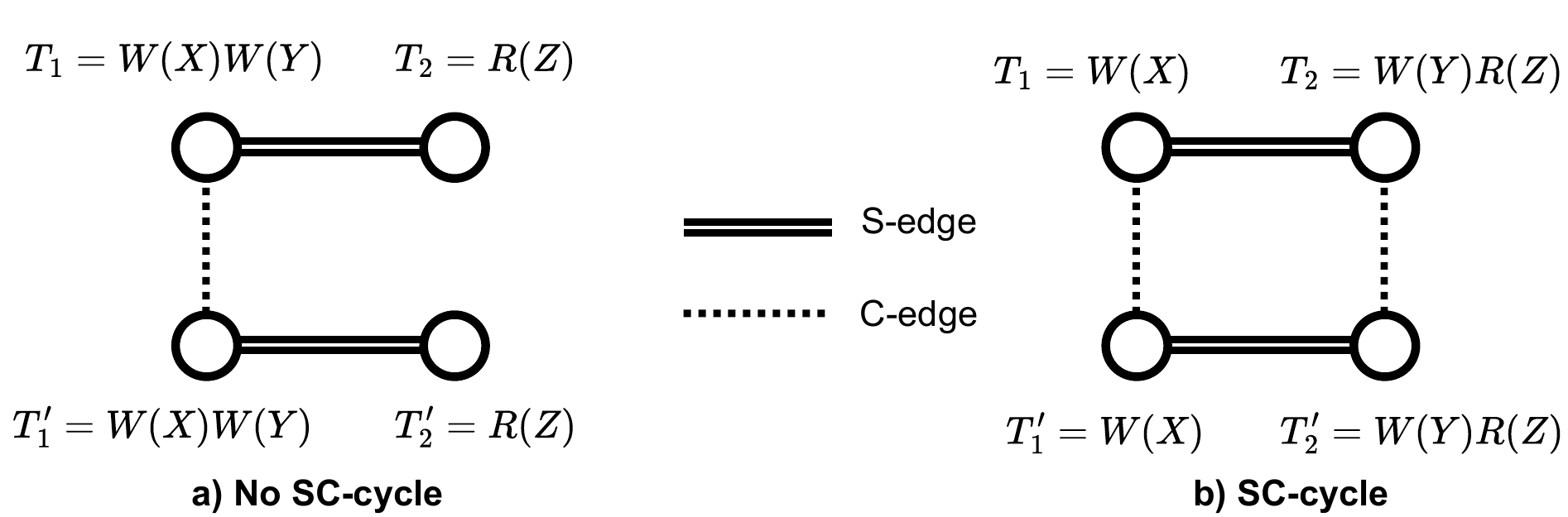}
    \vspace{-2.5em}
    \caption{SC-Graph of two different choppings for \( T = W(X)W(Y)R(Z) \): (a) Chopping into \( T_{1} = W(X)W(Y) \) and \( T_2 = R(Z) \), where the SC-Graph has no SC-cycle and is serializable; (b) Chopping into \( T_{1} = W(X) \) and \( T_2 = W(Y)R(Z) \), where the SC-Graph contains an SC-cycle and is not serializable.}

    \label{fig:sc-graph-background}
  \end{figure}

\cat{SC-Graph Representation}
In an SC-Graph, vertices represent subtransactions, and edges represent the relationships between these subtransactions. Vertices within the same chain are connected by (sibling) S-edges, while vertices in different chains that access the same data item, with at least one being a write operation, are connected by (conflict) C-edges.
\four{
In an SC-Graph, two instances of each transaction are included. This is necessary because the SC-Graph must capture scenarios where two separate, concurrent instances of the same transaction type conflict with each other. Including two instances in the SC-Graph allows these potential self-conflicts to be accurately modeled.
}

It has been proven that an SC-Graph without an SC-cycle (a cycle containing both S-edges and C-edges) guarantees serializability, even with piecewise execution~\cite{shasha1995transactionChopping}. 
Figure~\ref{fig:sc-graph-background}a shows the SC-Graph for subtransactions \( T_{1} = W(X)W(Y) \) and \( T_2 = R(Z) \). This chopping is serializable and valid as it does not create an SC-cycle. However, if the initial transaction \( T \) were instead chopped into \( T_{1} = W(X) \) and \( T_2 = W(Y)R(Z) \), an SC-cycle would be created, violating serializability, shown in Figure~\ref{fig:sc-graph-background}b.

The SC-graph of all involved transactions must have no SC-cycles, which limits the practicality of transaction chopping. This is because any practical application with more than a few transactions would typically contain an SC-cycle. One reason for this is that to have no SC-cycles, any pair of a decomposed transaction can only conflict on a single subtransaction.

\subsection{Why Deadlocks are Deadly?}

%
Deadlocks are problematic in the system for the following reasons.

\vspace{-1pt}
\cat{Deadlocks are Difficult to Handle}
Systems typically handle deadlocks through either prevention or detection. In deadlock prevention, systems abort transactions that could potentially create a deadlock, breaking the waiting cycle between transactions. However, this often results in aborting more transactions than necessary, leading to wasted computational resources and a loss of useful work. Additionally, most prevention mechanisms abort transactions mid-execution, which further wastes the progress made up to that point.

On the other hand, deadlock detection mechanisms present their own challenges, as they require the system to maintain a wait-for graph between transactions. This graph is costly to manage and imposes additional overhead on the system.

\vspace{-1pt}
\cat{Transaction Retrials are Problematic}
Deadlocks and transaction retrials can trigger a cascading effect, which exacerbates system instability and can lead to prolonged or even complete service outages~\cite{bronson2021metastable, huang2022metastable, habibi2024msf, habibi2024phd}. In a system with a continuous load of incoming transactions, retrials can accumulate, forming an artificial feedback loop that generates additional conflicts and further transaction retrials. This cycle intensifies contention within the system, resulting in degraded performance and work amplification.



\vspace{-5pt}

\section{\ourSystem\ Architecture}
\label{sec:Design}
In this section, we introduce the design of \ourSystem, which leverages our proposed graph representation, \graph\ (Section~\ref{sec:graph-defenition}), for the static analysis of transactions. 
Using \graph, \ourSystem\ achieves deadlock-free two-phase locking by preprocessing transactions, detecting potential sources of deadlocks, and providing strategies to eliminate them during execution (Section~\ref{sec:lock-manipulation}). 
Additionally, \ourSystem\ leverages a flexible transaction chopping~\cite{shasha1995transactionChopping} concept that we proposed, called \emph{partial transaction chopping} (Section~\ref{sec:partial-tx-chopping}), to enable the early release of locks without violating serializability.
\ourSystem\ provides a middleware algorithm (Section~\ref{sec:static-analysis}) to preprocess a static set of highly contentious transaction \five{types}. In \ourSystem, preprocessed transactions follow a basic 2PL protocol without requiring any deadlock-prevention mechanism. \ourSystem\ also supports the execution of dynamic transactions---transactions that have not been statically analyzed (Section~\ref{sec:DynamicTransactions}).

\vspace{-5pt}
\subsection{Intuition}

Assume there are only two types of transactions allowed in the system: \( T_1 = \text{W}(A) \text{W}(B) \) and \( T_2 = \text{W}(B) \text{W}(A) \). A possible deadlock scenario arises when an instance of transaction \( T_1 \) locks \five{a row of table \( A \)} while, concurrently, an instance of \( T_2 \) locks \five{a row of table \( B \)} to perform their respective write operations. \five{In this case, a deadlock may occur if both transactions attempt to access the same row held by the other.}
This happens because each transaction is waiting for the other to release its lock. One solution is to abort one of the transactions, but this leads to wasted CPU resources and does not prevent the same deadlock scenario from reoccurring.

A more effective solution is to ensure both transactions follow the same locking order for \five{tables \( A \) and \( B \)}. This approach eliminates the possibility of deadlocks and unnecessary aborts.

Additionally, deadlock-freedom enables early lock release, since transactions are guaranteed to commit without aborting from deadlocks. For example, in the scenario above, \( T_1 \) can release its lock on \five{row of table \( A \)} after the last operation on it, once it is certain the transaction will successfully commit.

However, in real-world systems, tracking all locks across many transactions is challenging. To address this, we propose \ourSystem, which enables the static analysis of transactions as a pre-processing step. By analyzing the \graph, potential deadlocks are identified, and elimination strategies are applied during a preprocessing phase. This ``abort-free'' property further allows for the early release of locks, reducing contention.
\ourSystem\ introduces \emph{partial transaction chopping} to fully exploit the benefits of early lock release. This technique builds upon the concept of early lock release and extends its potential by identifying additional opportunities within a transaction where locks can be safely released.


\vspace{-5pt}
\subsection{\ourSystem }

\ourSystem\ introduces a static pre-processing step for the two-phase locking (2PL) concurrency protocol. It converts a predefined set of transactions into a graph structure that is then analyzed prior to execution. This pre-processing analysis ensures deadlock-free operation during runtime, eliminating the need for transaction retrials and allowing for the early release of locks.

All statically analyzed transactions follow the 2PL protocol: locks are acquired during the growing phase, operations are executed, and locks are released during the shrinking phase. \ourSystem\ does not alter the core 2PL protocol for static transactions; rather, the changes lie in the order and timing of lock acquisition and release.
Systems can maintain their existing 2PL implementations while simply adjusting the order and timing in which locks are acquired and released. This allows for the seamless integration of our approach without modifying the underlying 2PL mechanism.

When dynamic transactions are enabled, they follow a modified 2PL to safely coexist with statically analyzed transactions. If only dynamic transactions are present, \ourSystem\ falls back to Wound-Wait 2PL. This design enables systems to retain their existing 2PL implementations while incorporating our minimal changes.

\subsection{\graph}
\label{sec:graph-defenition}
\graph\ represents transactions in a structured manner. Each node in the graph corresponds either to an operation or to a lock associated with an operation (referred to as a hop). Typically, a locking node precedes the operation nodes that modify the value of the locked variable. Within the same transaction, hops are linked sequentially by directed sibling edges (s-edges). When two locking hops from different chains conflict \five{by accessing the same table}, they are connected by undirected wait-for edges (w-edges).

%
A transaction is encoded as a sequence of hops ($[h_1 \cdots h_n]$), forming chains connected by s-edges. 
Additional transactions contribute further chains to the graph.

\cat{Representation} 
To formalize the \graph\ in graph-theoretic terms, for a set of transactions $T$, the corresponding \graph\ is defined as:
\( G = (V, E_s, E_w) \). Here, $V$ is the set of vertices or nodes, $E_s$ represents the set of directed sibling edges, and $E_w$ represents the set of undirected wait-for edges.

\cat{Hops}
Each vertex in the \graph\ represents a specific component of a transaction. The set of vertices $V$ is composed of four types of nodes:

\begin{enumerate}[leftmargin=*,topsep=0pt]
    \item Operation Nodes ($V_O$): Representing read/write operations within a transaction.
    \item Shared Locking Nodes ($V_{LS}$): Representing shared locking operations that precede read operations. 
    \item Exclusive Locking Nodes ($V_{LE}$): Representing exclusive locking operations that precede write operations.
    \item Unlock Nodes ($V_U$): Representing the nodes responsible for unlocking a set of locks held by the transaction. 
\end{enumerate}

Thus, the set of all locking nodes is defined as $V_L = V_{LE} \cup V_{LS}$, and the set of all vertices is defined as:
\( V = V_L \cup V_O \cup V_U \).

\cat{Sibling Edges}
Sibling edges ($E_s$) connect two nodes within the same transaction. If two hops $n$ and $n'$ are part of the same transaction, and one follows the other in execution order, they are connected via a directed s-edge.

\cat{Wait-for Edges}
Wait-for edges ($E_w$) represent potential conflicts in resource allocation between different transactions.
These edges connect two conflicting locking nodes associated with the same resource, where at least one of them is an exclusive lock. An undirected w-edge indicates a potential situation where one transaction waits for another to release a lock.

\five{
Static analysis cannot determine exactly what data items a hop accesses (which rows); therefore, we add a w-edge between two locking hops of different chains if the hops access the \emph{\textbf{same table}}.
Therefore, \ourSystem\ static analysis operates at the table level.
}
In \ourSystem\ implementation, programmers can provide annotations on the operations to specify that two nodes accessing the same table do not conflict, in which case the w-edge is removed. An example of such annotations is for operations that commute and, therefore, would not conflict even if accessing the same table.

\cat{Chains}
A chain is a sequence of hops \([h_1 \cdots h_n]\) (where \(h_i \in V\)), connected by s-edges and representing a transaction. Chains are interconnected via w-edges in the \graph.

\cat{SLW-cycle (Deadlocks)}
In the \graph, an SLW-cycle is defined as a \textit{directed} cycle that includes s-\textit{edges} ($E_s$), \textit{locking hops} ($V_{L})$, and \textit{w-edges}, with the constraint that no two w-edges can be adjacent. We prove (see Section \ref{proof}) that if no SLW-cycle exists in the \graph, deadlocks are impossible in runtime. Therefore, by eliminating SLW-cycles, we can guarantee a deadlock-free 2PL.

  \begin{figure}
    \centering
    \includegraphics[width=\linewidth]{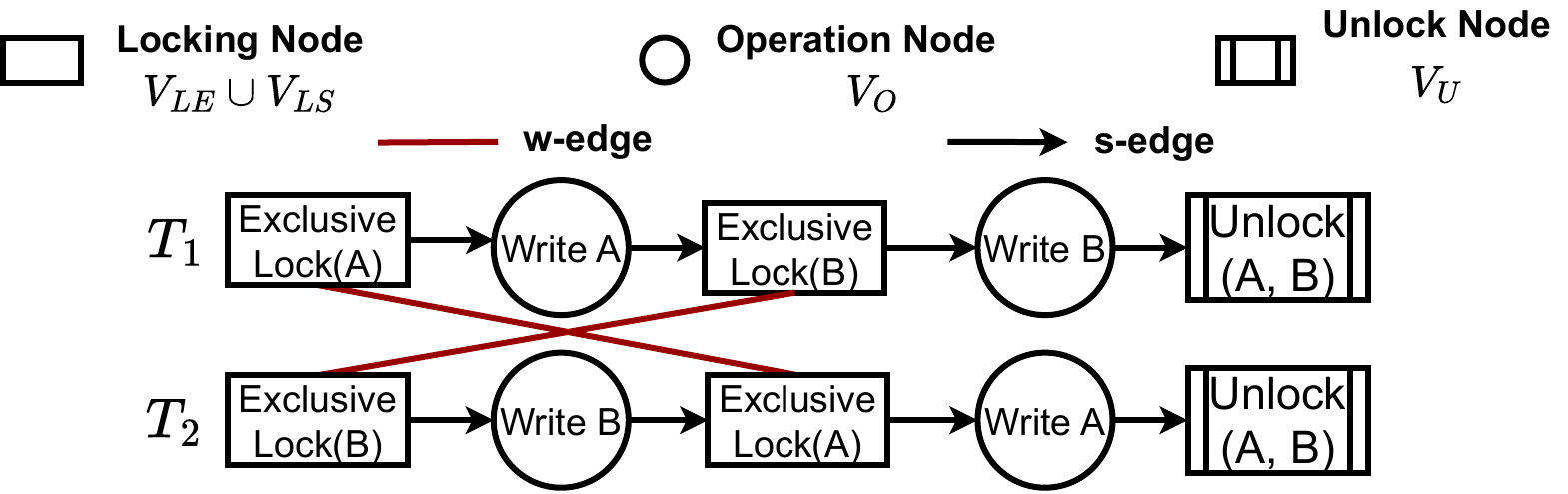}
    \vspace{-2em}
    \caption{\graph\ representation of $T_1= W(A) W(B)$ and $T_2= W(B) W(A)$ }
    \label{fig:slw}
  \end{figure}

\cat{Example}
Figure~\ref{fig:slw} shows an \graph\ representation of two transactions, $T_1= W(A) W(B)$ and $T_2= W(B) W(A)$. In this graph representation, there is an SLW-cycle, indicating a potential deadlock between these transactions where the lock for resource $A$ is held by $T_1$ and the lock for resource $B$ is held by $T_2$.


\begin{figure}
    \centering
    \includegraphics[width=0.85\linewidth]{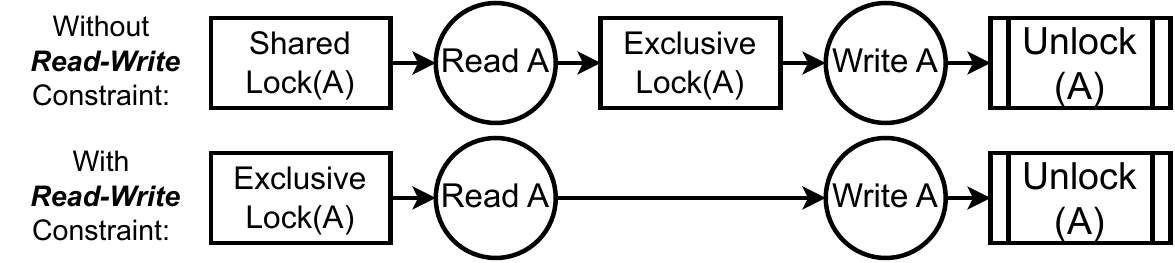}
    \vspace{-1em}
    \caption{\small The effect of the read-write constraint in \graph}
    \label{fig:slw-read-write}
\end{figure}

\cat{Read-Write Constraint}
Since \ourSystem\ performs static analysis on the \graph\ of transactions and can identify when a write follows a read, it provides an opportunity to combine both locks into a single node, reducing potential deadlocks caused by lock upgrades.
%
%
\four{In each transaction chain, if both read and write operations target the same resource \five{(same table)} within a single transaction type, \ourSystem\ applies a single exclusive lock for both in the \graph. This approach simplifies the graph structure and reduces the risk of deadlocks by avoiding shared lock conflicts.}


\begin{figure*}[t]
  \centering
    \begin{minipage}[b]{0.31\textwidth}
    \centering
    \includegraphics[width=0.9\linewidth]{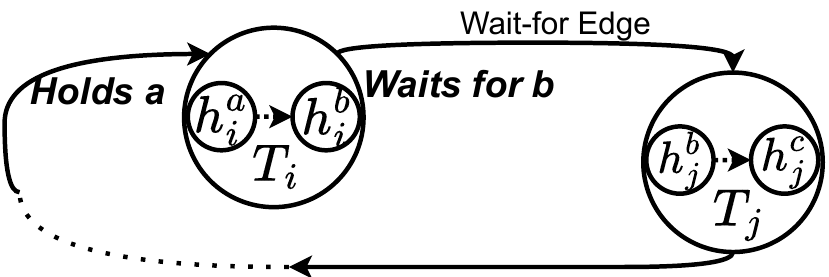}
    \vspace{-8pt}
    \caption{Breakdown of a potential wait-for-cycle resulting in an SLW-cycle}
    \label{fig:deadlock_prune_proof}
  \end{minipage}
      \hfill
   \begin{minipage}[b]{0.68\textwidth}
    \centering
    
    \begin{subfigure}[b]{0.55\textwidth}
        \centering
    \includegraphics[width=\linewidth]{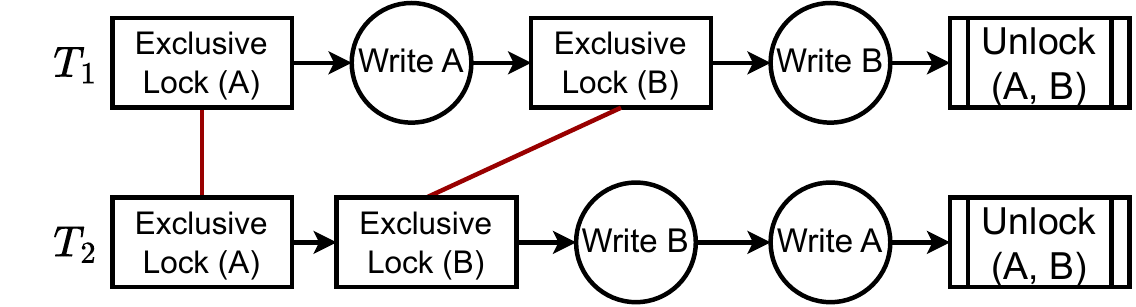}
    \caption{Lock movement strategy}
    \end{subfigure}
    \hfill
    \begin{subfigure}[b]{0.44\textwidth}
      \centering
    \includegraphics[width=\linewidth]{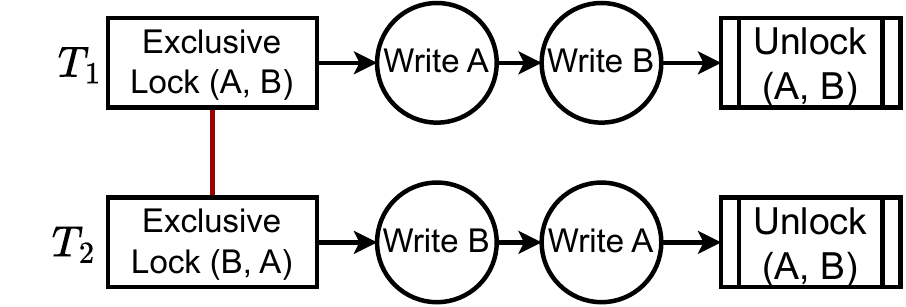}
    \caption{Lock combination strategy}
    \end{subfigure}
    \captionof{figure}{\small Lock manipulation techniques for eliminating SLW-cycles and preventing deadlocks}
    \label{fig:lock_manipulation}
   \end{minipage}

\end{figure*}

Figure~\ref{fig:slw-read-write} shows the effect of this constraint in the \graph\ of a transaction type \( T = \text{R}(A) \text{W}(A) \).
\four{
This transaction first performs a read on table A, followed by a write to \five{the same table}.  
In a typical 2PL execution, an instance of this transaction type would acquire a shared lock for the read and then upgrade to an exclusive lock before the write, often causing deadlocks.  
However, our static analysis detects this \texttt{R(A)→W(A)} access pattern replaces the locks with a single exclusive lock.
As a result, with this constraint applied, each instance of this transaction type acquires an exclusive lock on the resource before both the read and write operations.}

\vspace{-5pt}
\subsection{Correctness Proof}

\label{proof}

\begin{theorem}
A set of transactions is deadlock-free if, in their \graph, no SLW-cycle exists.
\end{theorem}

\begin{proof}
We will proceed by contradiction. Suppose that, despite the absence of an SLW-cycle, a deadlock occurs. Our goal is to show that this assumption leads to a contradiction, thus confirming that the transactions are deadlock-free if no SLW-cycle exists.

Let \( \{ T_1, T_2, \dots, T_n \} \) be the set of transactions involved in the deadlock. Under the 2PL protocol, a deadlock implies that these transactions form a cycle (a wait-for-cycle) in the wait-for graph, which can be described as:

\vspace{-1\baselineskip}
\[
\text{wait-for-cycle} = T_1 \rightarrow T_2 \rightarrow \dots \rightarrow T_n \rightarrow T_1,
\]

where each transaction \( T_i \) is waiting for a lock held by \( T_{i+1} \), with \( T_{n+1} = T_1 \).
Note that \( n \geq 2 \), as a single transaction cannot deadlock with itself---each transaction only locks a resource once (due to the read-write constraint).

In the \graph, each transaction \( T_i \) is represented as a sequence of hops connected by directed sibling edges (s-edges), reflecting the order of operations within the transaction. Let \( h_i^{s_1}, h_i^{s_2}, \dots, h_i^{s_n} \) be hops of \( T_i \), connected by s-edges:

\vspace{-1\baselineskip}
\[
h_i^{s_1} \xrightarrow[]{\text{s-edge}} h_i^{s_2} \xrightarrow[]{\text{s-edge}} \dots \xrightarrow[]{\text{s-edge}} h_i^{s_n}.
\]

We can construct an SLW-cycle using the transactions involved in the deadlock as follows:
\begin{enumerate}[leftmargin=*, topsep=0pt]
    \item Each wait-for edge in the wait-for-cycle corresponds to a w-edge in the \graph. In other words, if a wait-for edge exists between Transactions \( T_i \) and \( T_j \), then a \textbf{w-edge} exists between these transactions, connecting two locking hops (\( h_i^b \) and \( h_j^b \)) that conflict on the same resource, with \( h_i^b \) waiting for a lock held by \( h_j^b \).
   \item For each vertex \( T_i \) in the wait-for-cycle, there exist two distinct locking hops \( h_i^a \) and \( h_i^b \), connected by a directed \textbf{s-edge} (and possibly intermediate hops), where \( h_i^a \) occurs before \( h_i^b \). The hop \( h_i^a \) exists because a transaction is waiting for \( T_i \), and \( h_i^b \) exists as \( T_i \) is waiting for the next transaction in the wait-for-cycle. 
   Additionally, due to the existence of \( h_i^a \) and \( h_i^b \), w-edges are always separated by s-edges.
\end{enumerate}

Thus, an SLW-cycle is formed using this configuration, containing at least one \textbf{s-edge} (from step 2) and at least one \textbf{w-edge} (from step 1) connected by hops that are locking nodes ($V_{L}$). 
It is important to note that this cycle is directed, as each s-edge consistently leads from the hop that acquired the lock to the hop that is waiting for the lock. Additionally, in this cycle, no two w-edges are adjacent as they are always separated by s-edges (step 2).

This contradicts the initial assumption that no SLW-cycle exists in the \graph.
Therefore, the assumption must be false. Hence, if no SLW-cycle exists in the \graph, the set of transactions is deadlock-free under the 2PL protocol.
Figure~\ref{fig:deadlock_prune_proof} shows a breakdown of a wait-for-cycle leading to an SLW-cycle in the \graph.
\end{proof}

\subsection{Lock Manipulation Techniques}
\label{sec:lock-manipulation}
In 2PL, a lock must be acquired before its corresponding operation to ensure the serializability of transactions. However, there are no restrictions on the overall order of locks. \ourSystem\ utilizes this flexibility to eliminate deadlocks.
This section discusses the lock manipulation techniques that \ourSystem\ employs in its static analysis to eliminate potential deadlocks represented by an SLW-cycle.
This is achieved by repositioning certain locking nodes earlier in the transaction to create an ordered sequence of locks, merging locking nodes into atomic locking nodes, and annotating commutativity.


\cat{Moving Lock Nodes Earlier}
An effective deadlock elimination technique involves repositioning locking nodes earlier in the transaction sequence. Adjusting the timing of lock operations creates an ordered sequence for acquiring locks, transforming potential SLW-cycles into undirected cycles. Unlike SLW-cycles, undirected cycles do not lead to deadlocks.

Figure~\ref{fig:lock_manipulation}a demonstrates the removal of an SLW-cycle by moving a lock in \( T_2 \) earlier, thereby creating an ordered sequence of locks for resources \( A \) and \( B \). Even though there is a cycle in this \graph, it is not an SLW-cycle and will not lead to a deadlock.

Following the natural ordering of tables in the data model, particularly based on how foreign keys are defined in the relationships between tables, creates a consistent locking order for future transactions, ensuring that new transactions can also follow this ordering. 
However, there is no strict constraint on how locks should be ordered. As long as the SLW-cycle is eliminated, deadlocks will not occur during execution.
In our static analysis, \ourSystem\ explores various locking orders and selects the one that minimizes contention across transactions.



\cat{Combining Nodes}
One strategy involves combining multiple nodes in the \graph. This approach reduces the complexity of the graph by merging nodes, which leads to the removal of s-edges. By decreasing the number of s-edges, we eliminate the formation of SLW-cycles.
This method is particularly useful when lock acquisition of an operation is related to a previous operation, and it is not feasible to move any of the locks earlier. The atomic locking operation can be implemented using a secondary lock.
Figure~\ref{fig:lock_manipulation}b illustrates a possible combination of exclusive locks in transactions $T_1$ and $T_2$. To eliminate the sibling edge in the SLW-cycle, the locking of resource $B$ in $T_1$ and resource $A$ in $T_2$ are combined into an atomic locking node. This results in the removal of four sibling edges and resolves the cycle.

\five{
\ourSystem\ combines all locks on a table into a single locking node in the \graph\ for each transaction. Moving locks establishes a global deadlock-free locking order across tables, while atomic locks eliminate deadlocks within a single table.
}

\cat{Commutativity}
A w-edge between two resources can be removed if the resulting operations of those locks are known to be commutative. For instance, in many cases, two insert/delete operations are commutative, allowing the corresponding w-edge between these operations to be safely removed. Additionally, with application-level knowledge, certain operations can be explicitly annotated as commutative. This prevents the addition of a w-edge between their locking nodes, as they are non-conflicting.

\subsection{Partial Transaction Chopping}
\label{sec:partial-tx-chopping}
By releasing locks early, the level of parallelization increases. Knowing that transactions will not abort due to deadlocks enables other transactions to read resources written by uncommitted transactions, providing early write visibility. This parallelism is particularly beneficial for high-contention workloads where read-only transactions conflict with other transactions.

Releasing locks early for a resource divides a transaction into subtransactions. To maintain serializability, an SC-Graph must be created from these new subtransactions following lock release. This SC-Graph aligns with the transaction chopping concepts discussed in Section~\ref{sec:Background}.
\two{
However, traditional transaction chopping requires each sub-transaction to be fully \textbf{independent}---acquiring all locks at the start and releasing them upon completion---which often leads to SC-cycles and limits practicality.}
\two{
Our partial chopping relaxes this by allowing locks to be held across sub-transaction boundaries: a lock acquired early can be reused and released later in another sub-transaction. To enable this, we (1) redefine C-edges in the SC-Graph to connect conflicting \textit{locking nodes} (instead of operations), even if the operations occur across sub-transactions; (2) redefine sub-transactions by splitting at each unlock node.
%
%
%
This approach overcomes the inherent limitations of transaction chopping by providing a more flexible chopping mechanism without requiring alterations to the ordering of transaction operations.
}

In the following section, we discuss our approach for early lock release. 
We present our method for generating the SC-Graph, which enables conflict-serializable execution of subtransactions and differs from conventional SC-Graph generation in Transaction Chopping. We then prove the correctness of our method in ensuring serializability. For clarity, we retain the terminology of transaction chopping with some modifications in the SC-Graph generation.


\cat{SC-Graph Generation}
The SC-Graph is generated by chopping the transaction at each release of a lock, splitting the transaction into smaller sub-transactions based on unlock nodes. 
Subtransactions are connected using sibling edges (S-edge), and a subtransaction is connected to another with a conflict edge (C-edge) if they both have a conflicting locking node. As demonstrated in Section~\ref{sc-correctness}, conflicts between locking nodes are sufficient to ensure correctness, even if the operations occur in separate subtransactions.


\four{
To model potential self-conflicts, the SC-Graph includes two instances of each transaction type. This is necessary because, during execution, concurrent instances of the same transaction chain may access conflicting resources. Including two instances ensures that the SC-Graph captures these self-conflicts.
}

\begin{figure}
    \centering
    \includegraphics[width=\linewidth]{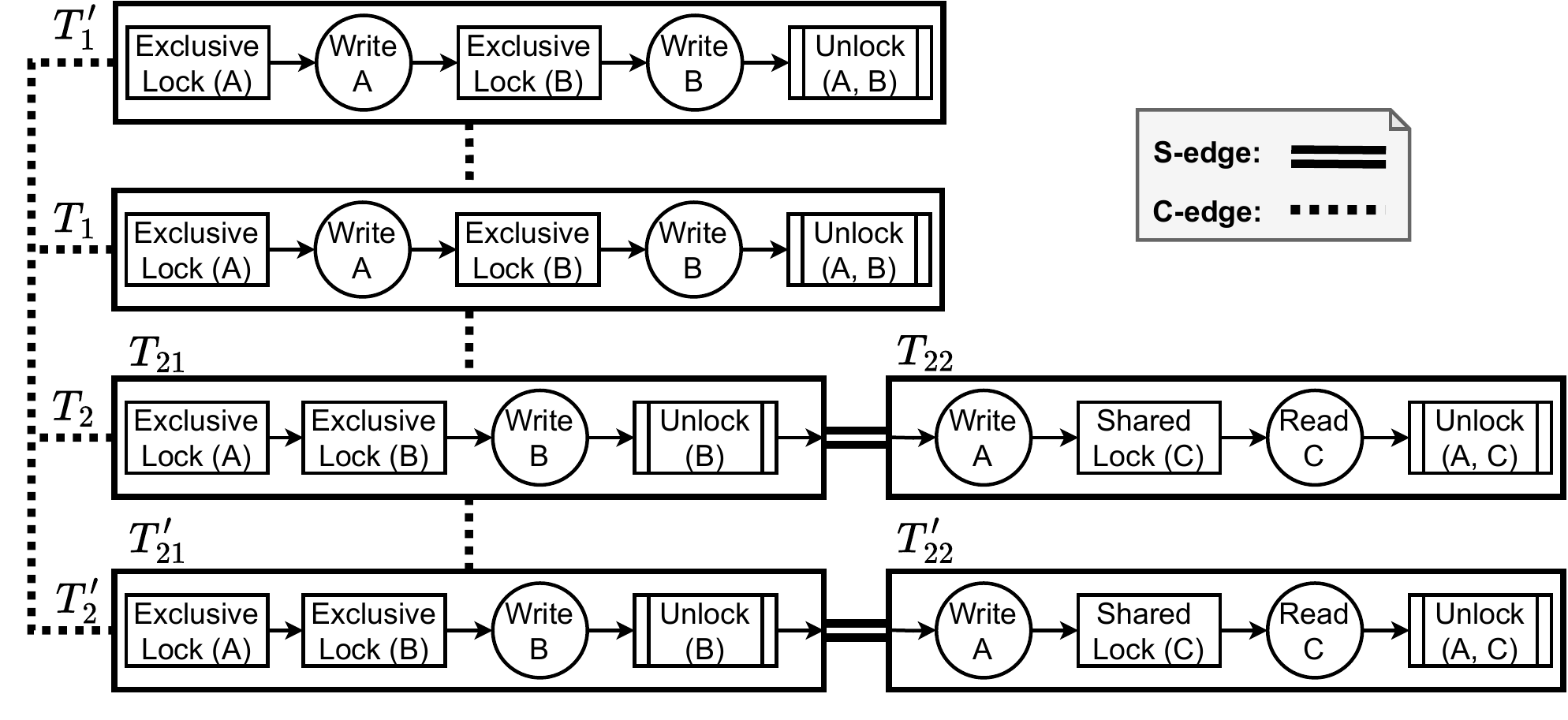}
    \vspace{-2em}
    \caption{The resulting SC-graph after chopping $T_2$ into two sub-transactions: $T_{21}=W(B)$ and $T_{22}=W(A)R(C)$}

    \label{fig:slw-sc-graph}
\end{figure}

\cat{Example}
\sloppy Assume the system has the two transactions $T_1= W(A) W(B)$ and $T_2= W(B) W(A) R(C)$, and we have already analyzed the transactions using \graph. Following the analysis, we decided to move the lock operation for resource $A$ earlier, before all locking nodes. By moving the lock earlier, we guarantee that no deadlock will occur during execution.

In the next step, to improve concurrency, we decide to release the lock on resource $B$ immediately after the write operation is completed, allowing other transactions to read its value. Figure~\ref{fig:slw-sc-graph} illustrates the resulting SC-graph after chopping $T_2$ into two sub-transactions: $T_{21}$ and $T_{22}$. 
Note that in this graph, two instances of each sub-transaction are included, as instances of the same chain might conflict. In this SC-graph, no SC-cycle exists, ensuring that the corresponding chopping remains serializable.


\cat{Correctness}
\label{sc-correctness}
\ourSystem\ enforces a specific ordering of operations through locks, meaning that subtransactions consistently follow the same order for conflicting locking nodes rather than the order of operations. This ensures that, in the Serialization Graph (SG), transactions (subtransactions) respect the direction of conflicts between locks rather than the sequence of operations. As a result, the SG between subtransactions is constructed based on resource locking rather than on individual operations.
Using this definition of SG, we will prove the following by contradiction. 

\begin{theorem}
Any execution of chopping where the SC-graph contains no SC-cycles yields an acyclic serialization graph on the given transaction instances $T_1, T_2, \dots, T_n$ and, therefore, is equivalent to a serial execution of committed transaction instances.
\end{theorem}

\begin{proof}
We proceed by contradiction. Consider an execution of chopping where the SC-graph contains no SC-cycle of $T_1, T_2, \dots, T_n$. Suppose that there exists a cycle in the serialization graph of $T_1, T_2, \dots, T_n$ resulting from this execution. That is, a cycle such as $T_l \rightarrow T_j \rightarrow \dots \rightarrow T_l$. Identify the subtransaction of the chopping associated with each transaction instance involved in this cycle: $ST_1 \rightarrow ST_2, \dots \rightarrow ST_n$, where both $ST_1$ and $ST_n$ belong to the same transaction instance $T_l$.

\begin{enumerate}[leftmargin=*,topsep=0pt]
    \item Sub transactions $ST_1$ and $ST_n$ cannot be identical since each piece adheres to two-phase locking by the execution rules, and it is a known fact that the serialization graph of a set of committed two-phase locked transaction instances is acyclic. Since $ST_1$ and $ST_n$ are distinct pieces within the same transaction instance $T_l$, there is an S-edge between them in the chopping graph.
    \item Every directed edge in the serialization graph cycle corresponds to a C-edge in the SC-graph, reflecting a conflict. Therefore, the cycle in the serialization graph implies the existence of an SC-cycle in the chopping graph, which contradicts our assumption of SC-acyclic execution.
\end{enumerate}
Thus, the execution of chopping with no SC-cycle in their SC-Graph guarantees conflict-serializability, making it equivalent to a serial execution of the transaction instances.

\end{proof}

\vspace{-12pt}
\cat{Self-Conflicts}
Self-conflicts may arise among subtransactions, leading to the creation of an SC-cycle between a transaction and a second instance of itself. 
An approach to eliminate self-conflict edges in later subtransactions is to move the conflicting locks to an earlier point, specifically to the first subtransaction. This adjustment removes the SC-cycle by preventing conflicts in later stages.
In \ourSystem's static analysis, we consistently move conflicting locks earlier to eliminate self-conflicting SC-cycles.




\subsection{Static Analysis}
\label{sec:static-analysis}

To statically analyze transactions, we propose an algorithm that takes a set of input transactions and produces a new set of transactions with modified lock ordering, while preserving the original operation ordering. 
The transformed transactions adhere to an SLW-cycle-free \graph\ and are further optimized using partial transaction chopping to release locks as early as possible.
The algorithm systematically explores all feasible lock acquisition strategies and selects the best one based on a contention ranking. The goal is to minimize contention among transactions, which requires a quantitative metric to evaluate and compare different strategies.

\cat{Contention Score}
To compare different \graphs\ of transactions, we define a contention score that measures the locking contention caused by each transaction.
The contention score is computed based on how long (how many operations) locks are held within a transaction, reflecting the potential blocking time for other transactions.
To account for the importance of each lock, we assign higher weights to locks on smaller tables, as they are considered more contentious.
We define the contention score as shown in Equation~\eqref{eq:contentionScore}, where higher contention scores indicate higher levels of contention within a given set of transactions.

\begin{equation}
\label{eq:contentionScore}
\text{Contention Score} = \footnotesize\sum_{\text{chain} \in \text{\graph}} \Biggl[
\quad \sum_{v \in V_{L}(\text{chain})}
\frac{\text{contention}(v)}{|\text{table}(v)|}
\Biggr]
\end{equation}

Equation~\eqref{eq:contentionScore} computes a contention score for the entire graph by summing, for each chain (representing a transaction) in the \graph, the contention levels of all locking nodes.
The inner summation iterates over all locking nodes \( v \in V_{L}(\text{chain}) \), where \( V_{L}(\text{chain}) \) denotes the set of locking nodes in the given chain. 
The function \( \text{contention}(v) \) calculates the number of operations between the locking node \( v \) and its corresponding unlocking node within the same transaction, implying that maintaining locks for longer periods increases the contention score. This value is then divided by \( |\text{table}(v)| \), the population (i.e., number of records) of the table associated with \( v \), accounting for the relative importance of each locking node. Tables with smaller populations create higher contention, thus increasing the overall contention score.

\cat{Static Analysis Algorithm} Algorithm~\ref{alg:static-analysis} provides an overview of \ourSystem\ static analysis of transactions.

\begin{algorithm}
\small
\caption{Brook2PL Static Analysis}
\label{alg:static-analysis}
\begin{algorithmic}[1]
\State \textbf{Input:} Set of transactions $\mathcal{T} = \{T_1, T_2, \dots, T_n\}$
\State \textbf{Output:} Deadlock-free transaction set $\mathcal{T}' = \{T'_1, T'_2, \dots, T'_n\}$
\Procedure{Brook2PL-Analyze}{$\mathcal{T}$}

    \State $G_{\text{init}} \gets \text{BuildInitialSLWGraph}(\mathcal{T})$
    \Comment{\footnotesize Build initial \graph\ including read-write constraints and commutativity---may contain SLW-cycles} \label{alg:static-analysis-init} \small

    \State $\mathcal{G}_{\text{deadlock-free}} \gets \text{GenerateCycleFree}(G_{\text{init}})$
    \Comment{\footnotesize Enumerate all SLW-acyclic graphs by reordering conflicting locking nodes using backtracking} \label{alg:static-analysis-backtrack} \small

    \State $\mathcal{G}_{\text{partially-chopped}} \gets \text{GreedyEarlyLockRelease}(\mathcal{G}_{\text{deadlock-free}})$
    \Comment{\footnotesize Greedily release locks earlier in \graph} \label{alg:static-analysis-partial-chopping} \small

    \State $G_{\text{Brook2PL}} \gets \text{SelectBestGraph}(\mathcal{G}_{\text{partially-chopped}})$
    \Comment{\footnotesize Select the best \graph\ based on the contention score} \label{alg:static-analysis-select} \small

    \State $\mathcal{T}' \gets G_{\text{Brook2PL}}.\text{ExtractTransactions}()$
    \label{alg:static-analysis-extract-transactions}

    \State \textbf{return} $\mathcal{T}'$

\EndProcedure
\end{algorithmic}
\end{algorithm}

First, Algorithm~\ref{alg:static-analysis} constructs an initial \graph\ of the input set of transactions in Line~\ref{alg:static-analysis-init} ($G_{\text{init}}$). This \graph\ follows the representation discussed in Section~\ref{sec:graph-defenition}, representing locking and operations as nodes and relations as edges. It also incorporates the read-write constraint of \graph\ and annotated commutativity to remove unnecessary w-edges. 
Note that the initial \graph\ of transactions might not be cycle free, thus leading to a potential deadlock during the execution.

Subsequently, \ourSystem\ employs a backtracking procedure in Line~\ref{alg:static-analysis-backtrack} to generate all feasible SLW-acyclic graphs by reordering conflicting locks earlier ($\mathcal{G}_{\text{deadlock-free}}$). 
This procedure works as follows: for each chain in \graph\ representing a transaction, the locking nodes are moved earlier in every possible way, generating all possible orderings of locking nodes within the chain. These combinations from all chains are then merged using the Cartesian product to generate every possible \graph\ that includes all feasible lock combinations. Note that only locks that conflict between chains (which might create a deadlock) are moved. The state space is also pruned by leveraging dependencies between operations: some operations (and their associated locks) depend on earlier operations in the chain due to transaction logic, and locking nodes cannot be moved earlier than their corresponding operations. Finally, only SLW-acyclic \graphs\ are returned as the set $\mathcal{G}_{\text{deadlock-free}}$.

After exploring all possible \graphs, Line~\ref{alg:static-analysis-partial-chopping} applies a greedy algorithm to release locks earlier and partially chop transactions, resulting in the set $\mathcal{G}_{\text{partially-chopped}}$. In each iteration, the algorithm moves an unlock node one position earlier in the transaction chain. 
During this process, we verify conditions for correctly partially chopping a transaction, as outlined in Section~\ref{sec:partial-tx-chopping}. 
The greedy algorithm also attempts to optimize the graph based on the contention score in Equation~\ref{eq:contentionScore}. If moving the unlock node one position earlier does not reduce this score, we stop moving the unlock node.

Finally, \ourSystem\ selects the best \graph\ with the lowest contention score (Line~\ref{alg:static-analysis-select}) and reconstructs the new deadlock-free transactions with modified lock ordering (Line~\ref{alg:static-analysis-extract-transactions}) as the output.
\five{
\ourSystem\ handles transactions with conditional logic by building a ''supergraph'' containing all possible execution paths (chains) of a transaction. The lock reordering algorithm then searches for all possible lock acquisition orders that are deadlock-free across all paths.
If a deadlock-free ordering for a transaction cannot be found, that transaction type is designated as ''dynamic'' and handled according to Section~\ref{sec:DynamicTransactions}.
}

\cat{Time Complexity}
The overall time complexity of Algorithm~\ref{alg:static-analysis} is dominated by the backtracking step (Line~\ref{alg:static-analysis-backtrack}), which explores all valid reorderings of conflicting locking nodes across transactions. For each transaction \(T_i\) with \(k_i\) movable locks, there are up to \(k_i!\) valid reorderings, leading to a worst-case complexity of 
\(
\mathcal{O}\left(\prod_{i=1}^{n} k_i!\right)
\)
for generating and validating all SLW-acyclic \graphs\ (where \(n\) is the total number of all transactions). The subsequent greedy optimization (Line~\ref{alg:static-analysis-partial-chopping}) attempts to move each unlock node earlier, taking up to \(\mathcal{O}(g \cdot u)\) time, where \(g\) is the number of generated acyclic \graphs\ and \(u\) is the number of unlocks per \graph. Finally, the \graph\ selection contributes \(\mathcal{O}(g \cdot \log g)\) time as we sort \graphs\ based on contention score. In practice, the search space is significantly reduced through pruning based on dependency constraints, making the approach tractable for typical workloads; moreover, the algorithm is highly parallelizable.

\vspace{-5pt}
\cat{Example}
\five{
We provide a concrete example of how transactions are converted into \graph\ and how static analysis, reorder, and combine locks to provide a deadlock-free version in Section~\ref{sec:online-game-store}.
}

\vspace{-15pt}
\subsection{Dynamic Transactions}
\label{sec:DynamicTransactions}

To handle dynamic transactions, \ourSystem\ categorizes transactions into two sets: the \(Static\) set, which includes contentious transactions that have been statically analyzed and are guaranteed never to abort, and the \(Dynamic\) set, which includes dynamic transactions that may abort in case of a conflict.
Our static analysis guarantees that transactions within the \( Static \) set cannot cause deadlocks among themselves. However, conflicting transactions involving the \( Dynamic \) set---either internally within the \( Dynamic \) set or between the \( Dynamic \) and \( Static \) sets---may still result in deadlocks.

In \ourSystem, we enable the optional execution of dynamic transactions by permitting aborts exclusively for those in the \( Dynamic \) set.
Our design assumes contentious transactions will be statically analyzed and included in the \( Static \) set. Consequently, introducing aborts only for dynamic transactions does not significantly increase abort rates, thereby avoiding degradation of system performance.

If dynamic transactions are enabled, we apply the following deadlock prevention protocol. When a transaction \( T_i \) requests a lock currently held by transaction \( T_j \); depending on the sets to which \( T_i \) and \( T_j \) belong, we follow one of these rules:

\begin{itemize}[leftmargin=*, topsep=0pt]
    \item \textbf{Both \( T_i \in Static \) and  \( T_j \in Static \):}
    \begin{itemize}[topsep=0pt]
        \item 
        No deadlock prevention is needed since this scenario does not lead to deadlocks.
    \end{itemize}

    \item \textbf{Both \( T_i \in Dynamic \) and \( T_j \in Dynamic \):}
    \begin{itemize}[topsep=0pt]
        \item 
        If \( T_i \) is \two{older} than \( T_j \), abort \( T_j \); otherwise, \( T_i \) waits.
    \end{itemize}

    \item \textbf{\( T_i \in Dynamic \) and \( T_j \in Static \):}
    \begin{itemize}[topsep=0pt]
        \item If \( T_i \) is \two{older} than \( T_j \), abort \( T_i \); otherwise, \( T_i \) waits.
    \end{itemize}

    \item \textbf{\( T_i \in Static \) and \( T_j \in Dynamic \):}
    \begin{itemize}[topsep=0pt]
        \item If \( T_i \) is \two{older} than \( T_j \), abort \( T_j \); otherwise, \( T_i \) waits.
    \end{itemize}
\end{itemize}

This protocol ensures a consistent timestamp-based ordering among all transactions. 
Note that \ourSystem\ falls back to Wound-Wait 2PL when only dynamic transactions are used. 
However, by statically analyzing transactions, we assign them higher priority over dynamic transactions in case of a conflict.

\five{Additionally, \ourSystem\ maintains correctness in case of uncommitted data reads when dynamic transactions are enabled. This is because static transactions are guaranteed not to have concurrency-control-related aborts. If a dynamic transaction reads uncommitted data from a static transaction and subsequently aborts, it will simply roll back its own changes with no effect on the static transaction, preserving correctness.}

\subsection{Limitations}
\ourSystem\ imposes certain limitations on the system, which we discuss in this section.

\cat{Preprocessing}
\five{
Similar to other preprocessing methods, our approach introduces some limitations due to the requirement of a preprocessing phase for transaction analysis.}
However, static analysis is performed offline as an optimization and is only necessary when new contentious transaction types are introduced to the system. It is not required during execution. \ourSystem's dynamic transaction handling approach is designed to support all workloads without the need for continuous static analysis.
\five{
\ourSystem\ prevents deadlocks in statically analyzed transactions. Deadlocks in dynamic transactions are unpredictable, but \ourSystem\ handles them to ensure the system remains correct. 
}

\five{
%
\ourSystem\ targets workloads that require prior knowledge of static transaction types. This aligns with real-world systems where many transactions are implemented as stored procedures or can be converted into static transactions with minimal changes~\cite{nguyen2025staticTransactions}.
}

\cat{User and Programmatic Aborts}
The early release of locks restricts the ability to perform aborts after locks are released in a static transaction. Once a lock on a resource is released, the resource may be read by other concurrent transactions, making it impossible to abort the previous operation without affecting serializability.
In \ourSystem, user-initiated aborts are only permitted before any locks are released in the static transaction. If a user wishes to abort a transaction after release of locks, they must provide a compensating transaction.
\four{For workloads with frequent user-initiated or data-dependent aborts, \ourSystem\ (and any protocol with early write visibility) faces limitations.}

However, if user aborts are necessary, besides disabling the early release of locks, a similar approach to previous works~\cite{guo2021bamboo} can be implemented by tracking dirty reads and performing cascading aborts in cases of inconsistency, \five{to benefit from early release of locks}.
\four{For workloads with a low rate of user aborts, this method will not be inefficient or impose a significant overhead on \ourSystem}.
\five{In all cases, \ourSystem\ guarantees deadlock-freedom.}
We omit the details of this method, as it is beyond the scope of this paper.

\begin{figure*}
    \centering
    \includegraphics[width=\linewidth]{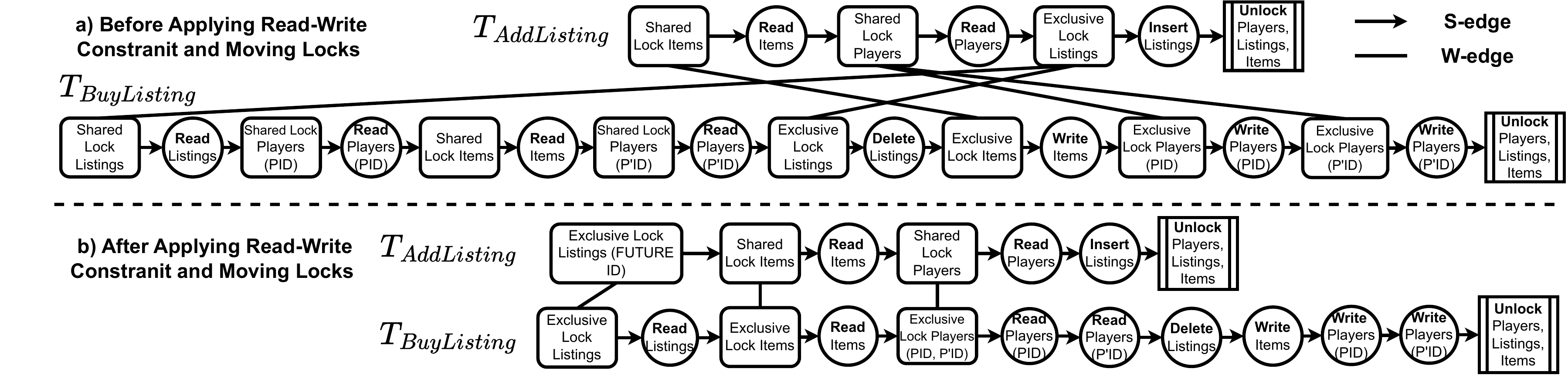}
    
    \caption{\graph\ representation of the store workload: (a) \graph\ without applying the read-write constraint or moving locks earlier; (b) \graph\ after applying the read-write constraint and moving locks earlier to ensure an SLW-cycle-free \graph. This execution will be utilized in \ourSystem.
}
    \label{fig:storeshop_slw}
\end{figure*}
\begin{figure}
    \centering
    \includegraphics[width=0.9\linewidth]{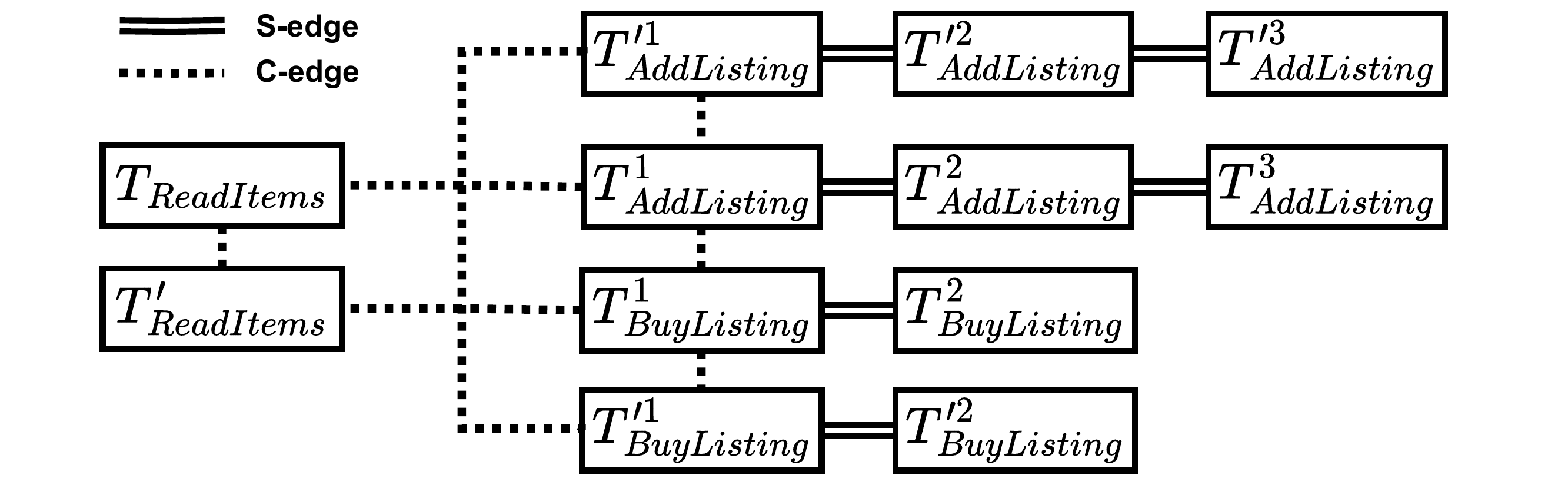}
    \vspace{-10pt}
    \caption{SC-Graph representation of the store workload after the early release of locks}
    \vspace{-5pt}
    \label{fig:storeshop_sc}
\end{figure}

\section{Evaluation}
\label{sec:Evaluation}
To evaluate \ourSystem, we perform several experiments comparing its performance against the state-of-the-art concurrency control protocol Bamboo~\cite{guo2021bamboo} as well as a basic 2PL protocol. The transactions used for evaluation fall into two categories. 
First, we use a synthetic transactional workload that simulates an online store for buying and selling game items.  This workload is designed to highlight various aspects of \ourSystem\ and to walk through the process of statically analyzing transactions.
Second, we evaluate \ourSystem\ using the TPC-C~\cite{tpc-c} benchmark under high-contention conditions, following the setup of prior works~\cite{ren2016deadlockfree, guo2021bamboo, prasaad2020handling}.

In the rest of this section, we first discuss the two workloads and their transactions used in our experiments and how \ourSystem\ statically analyzes them. Next, we describe the implementation of the transaction layer and our experimental setup. Finally, we present and analyze the results of various experiments focusing on transaction performance and tail latency.

\subsection{Online Game Store Workload}
\label{sec:online-game-store}
The store workload represents an online store for game items. The design of the data tables in this online game store includes three tables: \textbf{\textit{Players}}, \textbf{\textit{Items}}, and \textbf{\textit{Listings}}. The \textbf{Players} table contains information about individual players, including \texttt{PId} (a primary key), \texttt{name} (the player’s name), and \texttt{cash} (the player’s available balance). The \textbf{Items} table represents game items, with each record containing \texttt{IId} (a primary key), \texttt{name} (the item name), and \texttt{owner} (a foreign key referencing \texttt{PId} in the Players table to denote ownership). Finally, the \textbf{Listings} table tracks items available for sale, including \texttt{LId} (a primary key), \texttt{item} (a foreign key referencing \texttt{IId} from the Items table), and \texttt{price} (the listed sale price).


\begin{lstlisting}[float, caption={Add Listing Transaction}, floatplacement=t, label={alg:add_listing}]
Transaction AddListing(PId, IId, price): 
item = Read Item(IId) (*@\label{alg:add_listing_itemcheck}@*) // Check item exists and verify its owner
player = Read Player(PId) (*@\label{alg:add_listing_playercheck}@*) // Check player exists
Insert Listing(IId, price) (*@\label{alg:add_listing_insertListing}@*)
\end{lstlisting}

\begin{lstlisting}[float, caption={Buy Listing Transaction}, floatplacement=t, label={alg:buy_listing}]
Transaction BuyListing(PId, LId): 
listing = Read(Listings, LId) (*@\label{alg:buy_listing_listingCheck}@*) // Check listing exists
buyer = Read(Players, PId) // Check player & funds
item = Read(Items, listing.item) // Retrieve Item ID from Listing and verify item existence
owner = Read(Players, item.owner) (*@\label{alg:buy_listing_checkSeller}@*) // Check seller exists
Delete(Listings, LId) (*@\label{alg:buy_listing_deleteListing}@*) // Remove the listing after checks
Write(Items, listing.item) (*@\label{alg:buy_listing_transferOwner}@*) // Transfer ownership
Write(Players, PId) // Deduct buyer
Write(Players, item.owner) (*@\label{alg:buy_listing_creditMoney}@*) // Credit seller
\end{lstlisting}

\begin{lstlisting}[caption={Read Items Transaction}, label={alg:read_items}, float, floatplacement=t]
Transaction ReadItems(items):
for IId in items:
    Read Item(IId)
\end{lstlisting}
\label{sec:store-transactions}
\cat{Transactions}
Three transactions are implemented in this online game store. The first transaction, represented in Transaction~\ref{alg:add_listing}, shows how a player lists an item for sale with a specific price. In this transaction, after verifying the ownership of the item (line~\ref{alg:add_listing_itemcheck}) and ensuring the existence of the player (line~\ref{alg:add_listing_playercheck}), a new listing is inserted with the desired price.

The second transaction involves buying an already listed item, as shown in Transaction~\ref{alg:buy_listing}. 
In this transaction, we first perform initial checks to verify the existence of the listing and the player and ensure that the buyer has sufficient funds for the purchase (Lines~\ref{alg:buy_listing_listingCheck}–\ref{alg:buy_listing_checkSeller}). 
Next, the listing is deleted (Line~\ref{alg:buy_listing_deleteListing}), and the money is transferred from the buyer to the seller (Lines~\ref{alg:buy_listing_transferOwner}–\ref{alg:buy_listing_creditMoney}).


The third transaction, shown in Transaction~\ref{alg:read_items}, is a read-only transaction that retrieves a set of requested items. This transaction represents the potential read-only transactions that might be executed in a highly contentious system. Read-only transactions are significant because reads dominate real-world workloads~\cite{lu2020performanceOfReadOnly, bronson2013tao}.


\cat{Static Analysis}
We implement the static analysis algorithm shown in Algorithm~\ref{alg:static-analysis} to parse transactions and produce deadlock-free transactions. Figure~\ref{fig:storeshop_slw} shows the \graph\ representation for transactions in the store workload. 
In Figure~\ref{fig:storeshop_slw}a, the \graph\ is shown prior to applying the read-write constraint or moving any locks earlier—in other words, before \ourSystem\ static analysis.

Figure~\ref{fig:storeshop_slw}b shows the \graph\ after applying the read-write constraint and performing a lock movement. The read-write constraint reduces the locks used for Listings, Players, and Items into exclusive locks. Subsequently, by moving the Listings lock in the \textit{Add Listing} transaction to the beginning of the transaction, the SLW-cycle is eliminated, ensuring the deadlock-free execution of these transactions.
It is important to note that the \textit{Read Items} transaction is not shown in this figure, as the inclusion of a read-only transaction chain does not create any SLW-cycle in the \graph.

\cat{Partial Transaction Chopping}
Our static analysis algorithm improves concurrency by releasing the lock on \textit{Items} immediately after the final write operation in the \textit{Buy Listing} transaction. 
Additionally, it releases the locks on \textit{Items} and \textit{Players}, respectively, after the lock and read operations on \textit{Players} in the \textit{Add Listing} transaction.
The corresponding SC-Graph for this partial transaction chopping is shown in Figure~\ref{fig:storeshop_sc}. In this chopping, the \textit{Buy Listing} transaction is divided into two subtransactions, and the \textit{Add Listing} transaction into three. This SC-Graph does not contain any SC-cycles, ensuring the serializability of the final execution.

\subsection{TPC-C Workload}
The TPC-C benchmark~\cite{tpc-c} is widely used to evaluate concurrency control protocols under high-contention workloads. It simulates a wholesale supplier managing orders and payments across multiple warehouses in a dynamic and complex database environment.

Following previous works~\cite{ren2016deadlockfree, guo2021bamboo, prasaad2020handling}, we restrict our evaluation to TPC-C's \textit{New Order} and \textit{Payment} transactions, as they constitute the majority of the benchmark (approximately \(45\%\) and \(43\%\), respectively). Additionally, these transactions are short update transactions that place greater stress on concurrency control.
Therefore, in our evaluations, we conducted experiments with a workload comprising \(50\%\) \textit{New Order} transactions and \(50\%\) \textit{Payment} transactions (see Appendix A for a breakdown of their static analysis).
In accordance with the benchmark specifications, \(1\%\) of these transactions were aborted to simulate user aborts. In this benchmark, warehouses serve as hotspots for the workload, and contention increases as the number of warehouses decreases.

\vspace{-4pt}
\subsection{Implementation and Setup}
\vspace{-3pt}
\cat{System Implementation}
\ourSystem\ static analysis (Algorithm~\ref{alg:static-analysis}) is implemented in Java. This implementation parses the stored-procedure transactions---using the transaction representation language used in Section~\ref{sec:store-transactions}---and produces deadlock-free transactions as a preprocessing step. 
In our experiments, transactions, including their locks and operations, are written and executed as one-shot stored procedures, a common practice to evaluate concurrency control under high contention~\cite{ren2016deadlockfree, prasaad2020handling, guo2021bamboo, bailis2014coordination}. However, this does not preclude \ourSystem\ from handling interactive transactions, as dynamic transactions are supported.

Our 2PL implementation in Java uses a hash table as the lock table to store Lock objects. We also implement an in-memory database with Write-Ahead Logging (WAL), using a hash table to store the data. To enhance scalability, we use Java's ConcurrentHashMap~\cite{oracle_concurrenthashmap_java11}, which employs per-bucket locking that is highly optimized for concurrent execution. Additionally, our implementation acquires fine-grained (row-level) logical locks on records. Consequently, latch contention only occurs when multiple threads simultaneously acquire or release the logical locks of the same record.
Our implementation follows a basic 2PL that does not require conflict handling for static transactions, as they are designed to be deadlock-free. However, if dynamic transactions are enabled, conflicts are managed as described in Section~\ref{sec:DynamicTransactions}.

\vspace{-3pt}
\cat{Comparison Baseline}
We implement wound-wait 2PL~\cite{bernstein1987concurrency}, Bamboo~\cite{guo2021bamboo}, \re{optimistic concurrency control (OCC)~\cite{kung1981occ}, IC3~\cite{wang2016ics3}, and a basic lock sorting approach} as comparison baselines.



\begin{itemize}[leftmargin=*, topsep=0pt]
    \item \textbf{Wound-Wait:} The wound-wait variant of the 2PL protocol is a deadlock prevention mechanism. This mechanism ensures deadlock prevention while prioritizing older transactions.
    \item \textbf{Bamboo:} Bamboo is a state-of-the-art 2PL protocol that allows transactions to read dirty data during the execution phase (thereby violating 2PL) while still ensuring serializability. This protocol, which is a variant of wound-wait 2PL, permits a transaction to release its lock on a record after its last update, enabling other transactions to access the data. To enforce serializability, Bamboo tracks the dependencies of dirty reads and cascadingly aborts transactions when these dependencies are violated. 
    \two{Bamboo includes a latch-reduction optimization unrelated to concurrency control. Our implementation omits this to ensure a fair comparison focused solely on concurrency control logic.}
    \item \re{\textbf{Sorted Locks:} This 2PL baseline assumes known read/write sets (unlike \ourSystem) and acquires locks in sorted order at transaction start, ensuring deadlock-free execution.}
    \item \re{\textbf{OCC:} Optimistic concurrency control allows transactions to execute without locking resources, assuming conflicts are rare, and checks for conflicts before committing.}
    \item \re{\textbf{IC3:} IC3 is a state-of-the-art transaction chopping method that statically analyzes transactions and executes sub-transactions using OCC. It provides an approach to execute deadlock-prone SC-Cycles in SC-Graphs. We statically analyze our workloads using this method (see Appendix B). For a fair comparison, we apply our definition of conflicting operations in the static analysis: two operations conflict if they access the same row of a table, regardless of the column.} 
\end{itemize}


\vspace{-3pt}
\cat{Experimental Setup}
The experiments are conducted on \textit{cascadelake} node of Chameleon Cloud~\cite{keahey2020chameleon}. This machine has 2 \textit{Intel Xeon Gold 6242 CPUs} at \textit{2.8 GHz} with \textit{96} virtual threads and is also equipped with 187 GB RAM and 10 Gigabit Ethernet. 

The database size for our store workload is initialized with $500,000$ players in \textit{Players} table, each owning $5$ items, resulting in an \textit{Items} table containing $2,500,000$ rows. Additionally, the initial database includes $100,000$ pre-existing listings in the \textit{Listings} table.
%
We vary the database size in the TPC-C benchmark by increasing the number of warehouses.

To execute transactions, they are dispatched from a benchmark layer using \(64\) threads. Each benchmark runs for \(60\) seconds, excluding the warm-up period. In the event of a transaction abort, the transaction is immediately retried with its initial Timestamp ID.

\vspace{-2pt}
\cat{Controlling Contention}
To control the level of contention in our experiments, we use two metrics: 

\begin{itemize}[leftmargin=*, topsep=1pt]
    \item \textit{Number of Hot records (\( H \))} specifies the number of records in the hotspot. The hotspot records are selected from the tables frequently accessed by all transactions. For the store workload, \( H \) represents the number of Items and their associated Players chosen as hotspots. In the TPC-C workload, \( H \) warehouses are selected as hotspot records.
    
    \item \textit{Probability of hot record selection (\( P_{HotRecord} \))} defines the likelihood of selecting a hotspot record for the executing transaction. For instance, if the probability of hot record selection is \( 10\% \), there is a \( 10\% \) chance that the transaction's inputs will be chosen from hotspot records.
\end{itemize}

\two{These metrics create a controllable contention profile that allows us to systematically vary contention---unlike prior work, which often relies on fixed contention models~\cite{guo2021bamboo, wang2016ics3}.}

\subsection{Performance Comparison}
\begin{figure}[t]
    \centering
    \begin{subfigure}[b]{0.235\textwidth}
        \centering
        \includegraphics[width=\linewidth]{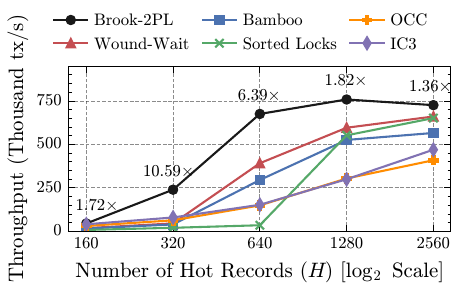}
        \caption{\re{\small Store Wrokload}}
        \label{fig:performance_store_varyingH}
    \end{subfigure}
    \hfill
    \begin{subfigure}[b]{0.235\textwidth}
        \centering
        \includegraphics[width=\linewidth]{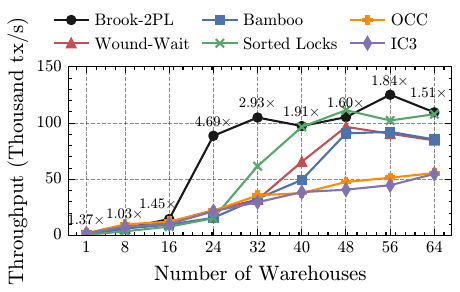}
        \caption{\re{TPC-C Workload}}
        \label{fig:performance_tpcc_varyingWarehouse}
    \end{subfigure}
     \vspace{-1\baselineskip}
    \caption{\re{\small \small Performance comparison while varying the number of hot records with \( P_{HotRecord} = 100\% \) }}
     \vspace{0.1\baselineskip}
     
    \label{fig:performance_varyingProb}
\end{figure}

\vspace{-2pt}
\cat{Varying Number of Hot Records (\( H \))}
We conducted experiments on the store workload by varying the number of hot records (\( H \)) to control the level of contention. In this experiment, the probability of selecting a hot record (\( P_{HotRecord} \)) is set to \( 100\% \). 
Figure~\ref{fig:performance_store_varyingH} illustrates the performance of the system using different 2PL protocols to execute two transactions: \textit{Add Listing Transaction} (Algorithm~\ref{alg:add_listing}) and \textit{Buy Listing Transaction} (Algorithm~\ref{alg:buy_listing}) in equal proportions.
Across all hotspot configurations, \ourSystem\ consistently outperforms competing protocols. 
\re{\ourSystem\ achieves an average speed-up in transaction processing of \( 4.3\times \) compared to others.}

\cat{Varying Number of Warehouses}  
In this experiment on the TPC-C workload, all warehouses are treated as hot records (\(P_{HotRecord} = 100\%\)), and the contention is increased by varying the number of warehouses (\(H\)). Figure~\ref{fig:performance_tpcc_varyingWarehouse} shows that under high contention (fewer warehouses---$H \leq 16$), all methods experience low performance, with \ourSystem\ outperforming both Bamboo and Wound-Wait. \re{As the number of warehouses exceeds $24$, \ourSystem\ approaches the system’s maximum capacity, continuing to outperform other retrial-based approaches, as many of their transactions are aborted.  For workloads with $H \geq 40$ and lower contention, our approach performs comparably to Sorted Locks.}  
This figure highlights the average speed-up of \ourSystem\ over Bamboo and Wound-Wait across different warehouse counts.


\begin{figure}[t]
    \centering
    \begin{subfigure}[b]{0.235\textwidth}
        \centering
        \includegraphics[width=\textwidth]{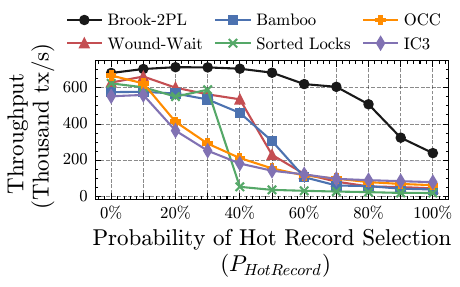} 
        \caption{\re{Store Workload ($H=320$)}}
        \label{fig:performance_store_varyingProb}
    \end{subfigure}
    \hfill
    \begin{subfigure}[b]{0.235\textwidth}
        \centering
        \includegraphics[width=\textwidth]{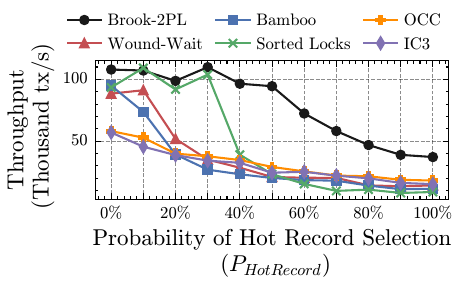} 
        \caption{\re{TPC-C Workload ($H=16$)}}
        \label{fig:performance_tpcc_varyingProb}
    \end{subfigure}
    \vspace{-1\baselineskip}
    \caption{\re{\small Performance comparison while varying the probability of hot record selection with a fixed number of hot records}}
    \vspace{-0.5\baselineskip}
    \label{fig:performance_varyingProb}
\end{figure}
\cat{Varying Probability of Hot Record Selection (\(P_{HotRecord}\))}
This experiment adjusts the contention of the workload by varying \(P_{HotRecord}\) while keeping the number of hot records (\(H\)) fixed.
Figure~\ref{fig:performance_store_varyingProb} shows the performance of comparing 2PL protocols as \(P_{HotRecord}\) varies from \(0\%\) (no contention) to \(100\%\) (all records are selected from hot records) for the \textit{Add Listing} and \textit{Buy Listing} transactions. In this experiment, \(320\) \textit{Item} records and their corresponding \textit{Players} are selected as hotspot records.

The results show that when there is no contention, all protocols perform similarly. 
Bamboo performs slightly worse than Wound-Wait in this scenario due to its additional overhead of tracking dirty reads. 
\re{IC3 does not provide any performance gain over OCC, as its static analysis of the Store workload cannot chop transactions into multiple hops due to deadlocks (each transaction is combined into a single hop). Additionally, IC3 incurs overhead from writing to a local stash for each hop and introduces extra waiting, making it less efficient.}
As contention increases, the performance of other protocols begins to decline, whereas \ourSystem\ maintains stable performance under lower levels of contention.
\re{For \(P_{HotRecord} < 40\%\), the throughput of the Sorted Locks approach does not drop significantly. However, for higher contention levels, performance drops drastically due to long waiting times. A simple sorting approach can create bottlenecks under high contention by forcing all transactions to acquire the same hot locks first. Also, moving all lock acquisitions to the beginning of the transaction forces locks to be held for a longer duration, reducing concurrency.}

\begin{figure}[t]
    \centering

    \begin{subfigure}[b]{0.235\textwidth}
        \centering
        \includegraphics[width=\textwidth]{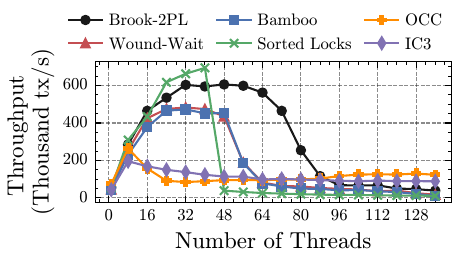} 
        \caption{\re{\scriptsize Store [\( P_{HotRecord} = 70\% \) \(, H=320 \)]}}
        \label{fig:performance_varyingThreads_store}
    \end{subfigure}
    \hfill
    \begin{subfigure}[b]{0.235\textwidth}
        \centering
        \includegraphics[width=\textwidth]{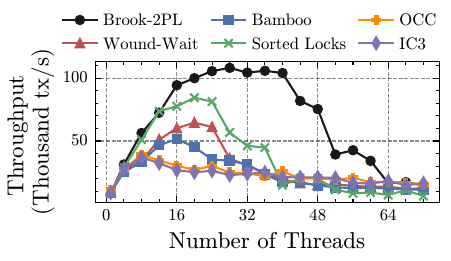} 
        \caption{\re{\scriptsize TPC-C [\( P_{HotRecord} = 100\% \) \(, H=16 \)]}}
        \label{fig:performance_varyingThreads_tpcc}
    \end{subfigure}
    \caption{\re{Performance comparison while varying the number of threads with fixed $H$ and $P_{HotRecord}$}}
    \label{fig:readonly}
\end{figure}

Interestingly, when the probability of hot records is low (\(P_{{HotRecord}} < 50\%\)), \ourSystem\ shows a slight performance improvement. This is because hot records fit within the cache, resulting in a higher cache hit rate and reduced access latency.
\ourSystem\ tolerates contention up to \(50\%\) without a performance drop. 
Overall, \ourSystem\ shows less performance degradation than the comparison methods. In this experiment, \ourSystem\ achieves an average \re{\(5.2\times\) speed‑up} over the other methods.

A similar experiment is conducted on the TPC-C workload (Figure~\ref{fig:performance_tpcc_varyingProb}). In this experiment, the database contains \(64\) warehouses, with \(16\) warehouses selected as hot records (\(H = 16\)). Similar to the experiments on the store workload, \ourSystem\ outperforms the comparison methods and shows lower performance degradation. 
\re{The throughput of the Sorted Locks approach declines sharply under high contention  (\(P_{{HotRecord}} > 30\%\)), whereas it remains comparable to \ourSystem\ at lower contention levels.}
Bamboo and Wound‑Wait perform similarly under high contention, since Bamboo’s lock retirement mechanism only retires exclusive locks and does not provide significant performance gains in the TPC-C workload.
\re{OCC and IC3 consistently underperform compared to our method across all contention levels, as many transactions abort due to conflicts. Even with \(P_{{HotRecord}} = 0\%\), conflicts exist between transactions in the TPC-C workload. OCC-based approaches are generally unsuitable for high-contention workloads due to their high abort rates.}
\ourSystem\ achieves an average speed-up of \re{\(2.86\times\)}.

\begin{figure}[t]
    \centering
    \includegraphics[width=0.8\linewidth]{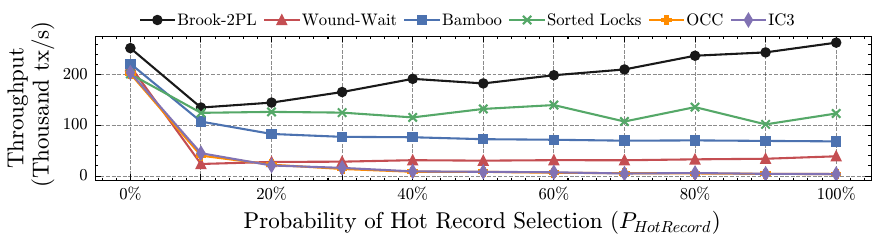} 
    \vspace{-1\baselineskip}
    \caption{\re{Performance comparison of read-only transactions while varying the probability of hot record selection ($H = 320$)}}
    \label{fig:readonly}
\end{figure}

\cat{Varying Number of Threads}  
This experiment evaluates scalability by increasing the number of threads under a high contention setting.
Figure~\ref{fig:performance_varyingThreads_store} shows the performance for the store workload with high contention (\(P_{HotRecord} = 70\%\), \(H = 320\)). While all methods perform similarly at lower thread counts, Bamboo, Wound-Wait, \re{and Sorted Locks} suffer performance drops after $48$ threads.
\re{OCC and IC3 fail to scale under highly contended workloads, reaching their peak performance with only $8$ threads.}
\ourSystem\ maintains stable and higher throughput, showcasing better scalability under high contention.
Figure~\ref{fig:performance_varyingThreads_tpcc} shows results for TPC-C with high contention (\(P_{HotRecord} = 100\%\), \(H = 16\)). As the number of threads increases, all methods initially improve in throughput, \re{but OCC and IC3 degrade beyond $8$ threads, Bamboo and Wound-Wait beyond $16$ threads, and Sorted Locks after $20$ threads.} In contrast, \ourSystem\ continues to scale up to $40$ threads, sustaining higher throughput and maintaining scalability. 

\begin{figure}[t]
    \centering
    \begin{subfigure}[b]{0.245\textwidth}
        \centering
        \includegraphics[width=\textwidth]{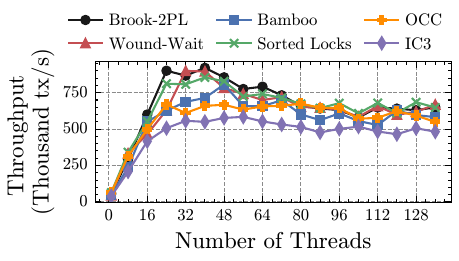} 
        \caption{\two{\scriptsize Store [\( P_{HotRecord} =0\% \) \(, H=320 \)]}}
        \label{fig:performance_varyingThreads_0p_store}
    \end{subfigure}
    \hfill
    \begin{subfigure}[b]{0.22\textwidth}
        \centering
        \includegraphics[width=\textwidth]{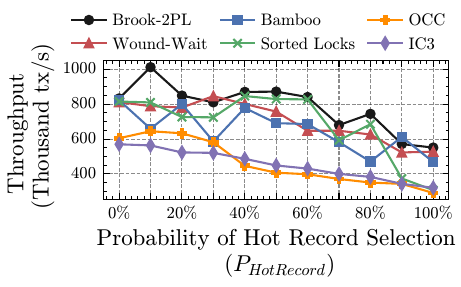} 
        \caption{\two{\scriptsize Store [\( 32 \) Threads \(, H=320 \)]}}
        \label{fig:performance_store_varyingProb_32threads}
    \end{subfigure}
    \caption{\two{Performance comparison of protocols with lower levels of contention}}
    \vspace{-5pt}
    \label{fig:thread_count}
\end{figure}

\vspace{-1.5pt}
\cat{Read-Only Transactions}  
To evaluate the effect of early release of locks on transactions, we execute the \textit{Read Items} transaction--which reads $20$ hot items in the store workload--in \(20\) additional threads alongside the \textit{Add Listing} and \textit{Buy Listing} transactions. 

Figure~\ref{fig:readonly} shows the performance of read-only transactions under varying \(P_{HotRecord}\).
\ourSystem\ demonstrates higher throughput than the other methods across all values of \(P_{HotRecord}\). Under no contention, all methods perform similarly. However, Bamboo's and Sorted Locks throughput drops after \(P_{HotRecord} = 10\%\), leveling off near $80\text{k}$ and $120\text{k}$ transactions/sec respectively beyond \(P_{HotRecord} = 20\%\). In contrast, \ourSystem\ maintains stable throughput, with at most a \(53\%\) drop in the worst case.

\ourSystem\ demonstrates higher throughput than the other methods across all values of \(P_{HotRecord}\). Under no contention, all methods perform similarly. However, \re{Bamboo's and \re{Sorted Locks'} throughput drops after \(P_{HotRecord} = 10\%\), leveling off near $80\text{k}$ and $120\text{k}$ transactions/sec, respectively, beyond \(P_{HotRecord} = 20\%\).} In contrast, \ourSystem\ maintains stable throughput, with at most a \(53\%\) drop in the worst case.
\ourSystem's throughput increases under higher contention, as more hot records are cached, reducing the average data retrieval time from the in-memory database. 
\re{The throughput of other methods remains close to zero under high contention, as they lack mechanisms to support long read-only transactions and suffer from high abort rates.}

\vspace{-1.5pt}
\cat{Choice of 64 Threads}
\two{We selected 64 threads for our experiments to (1) saturate all cores and (2) match the thread-to-warehouse ratio in TPC-C, following standard practice.
Figure~\ref{fig:performance_varyingThreads_0p_store} confirms that under $0\%$ hot record probability, all protocols scale well to $64$ threads without thrashing. 
Figure~\ref{fig:performance_store_varyingProb_32threads} shows the performance of protocols under varying \(P_{HotRecord}\) with $32$ threads. While trends remain similar, performance gaps are less pronounced, supporting the need for higher thread counts to generate sufficient contention and effectively differentiate protocol performance.
}

\begin{figure}[t]
    \centering
    \begin{subfigure}[b]{0.235\textwidth}
        \centering
        \includegraphics[width=\textwidth]{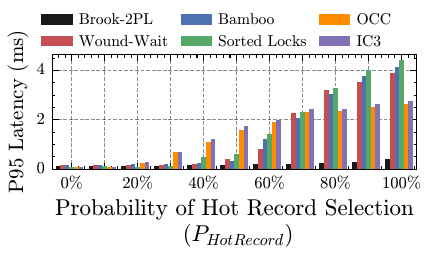} 
        \caption{Store Workload ($H=320$)}
        \label{fig:latency_store}
    \end{subfigure}
    \hfill
    \begin{subfigure}[b]{0.235\textwidth}
        \centering
        \includegraphics[width=\textwidth]{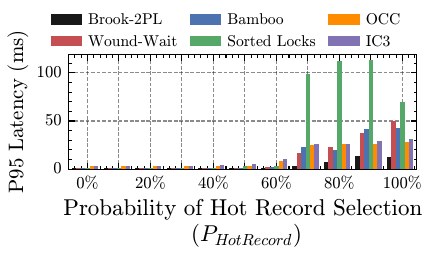} 
        \caption{TPC-C Workload ($H=16$)}
        \label{fig:latency_tpcc}
    \end{subfigure}
        \vspace{-12pt}
        \caption{\re{Tail Latency comparison while varying the probability of hot record selection with a fixed $H$}}
        \vspace{-5pt}
    \label{fig:latency}
\end{figure}

\vspace{-1.5pt}
\subsection{Tail Latency Comparison}
Figure~\ref{fig:latency} shows the tail latency (p95) of transactions as \(P_{HotRecord}\) varies. Latency is measured from the time a transaction is submitted to the system until it is committed. For protocols with aborts, transactions are retried immediately after being aborted. In contrast, \ourSystem\ experiences no aborts due to deadlocks, and each transaction is executed only once.

Figure~\ref{fig:latency_store} shows the p95 latency of transactions when \(H = 320\) in the store workload. \ourSystem\ consistently achieves lower tail latency. Bamboo suffers from cascading aborts, resulting in higher latency than Wound-Wait across most contention levels. 
\re{IC3 and OCC achieve similar latency; however, IC3 has slightly higher tail latency due to additional overhead. The Sorted Locks approach achieves the lowest latency under low contention but experiences an increase in latency under high contention, ultimately performing the worst.}
On average, \ourSystem\ reduces p95 latency \re{by \(56\%\) compared to other baselines}.

Figure~\ref{fig:latency_tpcc} shows a similar experiment for the TPC-C workload with \(H = 16\) out of \(64\) warehouses. Tail latency for other protocols increases rapidly due to high abort rates; both exhibit similar latency patterns under this setup. In contrast, \ourSystem\ provides significantly better p95 latency as contention increases.
\re{On average, \ourSystem\ reduces tail latency (p95) by \(48\%\) compared to other protocols in the TPC-C workload.}

\begin{figure}[t]
    \centering
    \includegraphics[width=\linewidth]{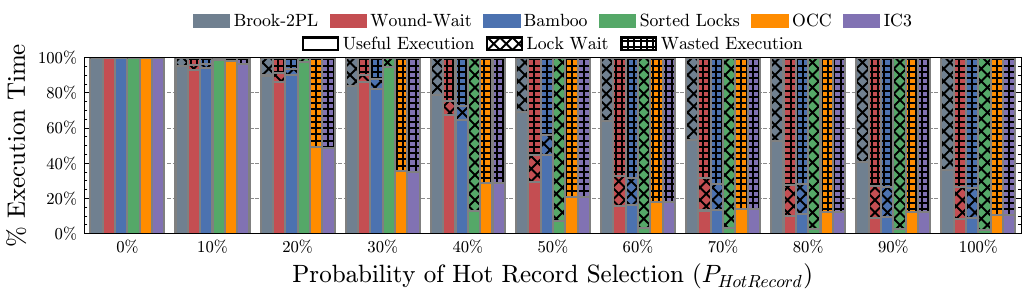}
    \vspace{-2.5em}
    \caption{\re{CPU time breakdown--store workload (\(H = 320\))}}
    \vspace{-1.5pt}
    \label{fig:work-prportion}
\end{figure}

\vspace{-3pt}
\subsection{Execution Time Breakdown}

Figure~\ref{fig:work-prportion} shows the breakdown of CPU time—useful execution, lock wait, and wasted execution—for the \textit{Add Listing} and \textit{Buy Listing} transactions in the store workload as \(P_{HotRecord}\) increases. Useful execution represents the proportion of transaction execution time spent performing operations, lock wait indicates the proportion of time spent waiting for locks, and wasted execution corresponds to the proportion of transaction work wasted due to aborts.

\ourSystem\ experiences no wasted execution because no transactions are aborted. As contention increases, the proportion of lock wait for \ourSystem\ increases; however, it remains lower than the wasted execution observed in the comparison methods. 
\re{The lock wait time in the Sorted Locks approach increases drastically under high contention, consuming nearly the entire execution time.}
Both Bamboo and Wound-Wait experience significant wasted execution in addition to lock wait.
\re{IC3 and OCC experience significantly increased wasted execution time as contention increases.}
In particular, on average, \ourSystem\ spends \(70\%\) of the execution time performing useful work, \re{while other methods spend \(41\%\).} 


\begin{figure}[t]
    \centering
    \begin{subfigure}[b]{0.24\textwidth}
        \centering
        \includegraphics[width=\textwidth]{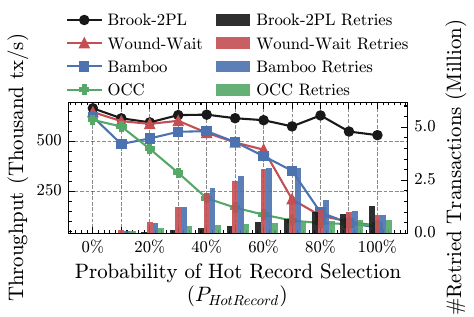} 
        \caption{$10\%$ Dynamic Transction}
        \label{fig:performance_store_dynamic}
    \end{subfigure}
    \hfill
    \begin{subfigure}[b]{0.23\textwidth}
        \centering
        \includegraphics[width=\textwidth]{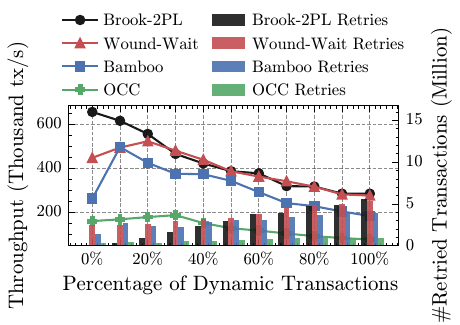} 
        \caption{$P_{HotRecord} = 50\%$}
        \label{fig:performance_store_dynamic_percentage}
    \end{subfigure}
    \vspace{-1\baselineskip}
    \caption{\re{Performance and Retry Analysis with Dynamic Transactions enabled in the Store Workload}}
    \label{fig:performance_dynamic}
\end{figure}

\vspace{-5pt}
\subsection{Dynamic Transactions}  
To evaluate dynamic transactions, we enable dynamic transaction support in \ourSystem\ and introduce transactions that perform \(10\) read–modify–write operations. 
Dynamic transaction inputs are selected from database records, with each input drawn from hot records with probability~\(P_{{HotRecord}}\).
\re{IC3 and Sorted Locks are excluded as they do not support dynamic transactions.}

Figure~\ref{fig:performance_store_dynamic} shows the throughput (left axis) and the number of transactions being retried (right axis) as \(P_{HotRecord}\) varies for both the store workloads when $10\%$ transactions are dynamic and the rest are from previous static transactions.
\ourSystem\ maintains higher throughput across all contention levels. Although dynamic transactions introduce retries, \ourSystem\ shows significantly fewer retrials. Bamboo and Wound-Wait suffer from increases in retried transactions, particularly beyond \(P_{HotRecord} = 40\%\), leading to noticeable throughput degradation. Beyond \(P_{HotRecord} = 60\%\),  retry counts decrease not because performance improves but because excessive waiting times prevent many transactions from attempting a retry.
\re{For OCC, long dynamic transactions cause most transactions in the system to be continuously retried, significantly reducing overall performance.
}
%

In Figure~\ref{fig:performance_store_dynamic_percentage}, we vary the share of dynamic transactions from \(0\%\) to \(100\%\) with \(P_{{HotRecord}}=50\%\).
When only static transactions are present, throughput of \ourSystem\ outperforms competing methods and experiences no retries, since it never aborts static transactions.
Introducing $10\%$ dynamic transactions increases workload diversity and introduces more transaction types, giving Wound‑Wait and Bamboo a throughput boost. 
However, as the share of dynamic transactions increases higher than $10\%$, throughput steadily declines for all methods. 
\ourSystem\ maintains its advantage up to $50\%$ dynamic transactions, beyond which its throughput converges with Wound‑Wait’s. Bamboo consistently has lower throughput, because it immediately releases the lock after every write in dynamic transactions---without knowing whether a write is the last update---resulting in extra overhead.
\re{OCC experiences low throughput across all levels of dynamic transactions due to a high abort rate.}
\ourSystem\ achieves an average performance improvement of \re{$95\%$} on static transactions and \re{$84\%$} on dynamic transactions.
\ourSystem\ experiences fewer transaction retries until dynamic transactions reach $60\%$, after which its retry count is comparable to other methods. Overall, \ourSystem\ has $21\%$ fewer transactions retried than Bamboo and $33\%$ fewer than Wound‑Wait.
\ourSystem\ benefits from a higher percentage of static contentious transactions. Dynamic transactions that become highly contentious within the workload are targeted for static analysis. 

\begin{figure*}[t]
    \centering
    \includegraphics[width=0.85\textwidth]{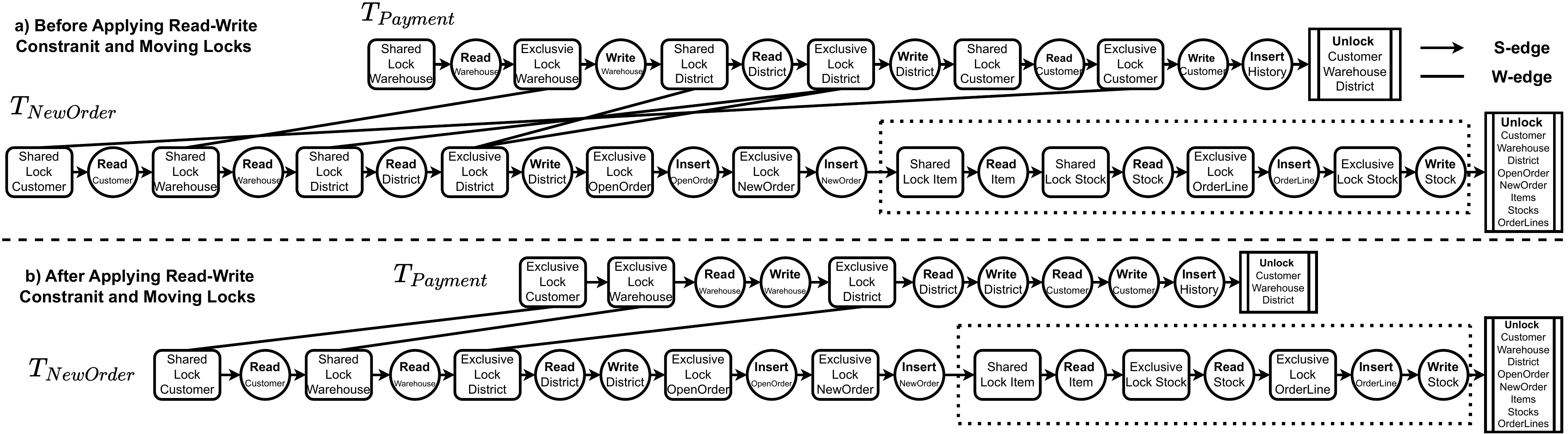}
    
    \caption{\graph\ representation of TPC-C workload transactions: (a) \graph\ without applying the read-write constraint or moving locks earlier; (b) \graph\ after applying the read-write constraint and moving locks earlier.
}
    \label{fig:tpcc_slw}
\end{figure*}
\vspace{-2pt}
\section{Related Work}
\label{sec:Related_work}
\vspace{-2pt}
\cat{Deadlock and Abort Freedom}
Many research efforts have aimed to develop deadlock-free concurrency control protocols, often using deadlock-prevention mechanisms. They inevitably abort and retry conflicting transactions to break dependency cycles~\cite{agrawal1995ordered, guo2021bamboo, bernstein1987concurrency, kimura2012efficientLocking, tu2013speedy, zheng2014fast, levandoski2015multi, lim2017cicada, mu2014extracting, yuan2016bcc}. 
However, these approaches may unnecessarily abort transactions, reducing concurrency and potentially creating a cascading burden on the system.

Deterministic systems, instead of applying concurrency control to individual transactions, process transactions in batches and execute them according to pre-generated schedules. These systems process transactions in a serial order based on their logged read and write sets, which are predetermined to generate optimal solutions. They ensure serializable execution while avoiding concurrency-control-related aborts~\cite{qin2021caracal, faleiro2017earlyWriteVisibility, cowling2012granola}.
QueCC~\cite{qadah2018quecc} proposes a mechanism to avoid transaction aborts using a priority queue approach to dispatch transactions into different partitions. Similarly, Calvin and Bohm~\cite{faleiro2014bohm, thomson2012calvin} utilize knowledge of transaction conflicts to relax the total order into an equivalent partial order, thereby reducing transaction aborts. Related ideas are applied in other deterministic transaction systems~\cite{ding2018improving, narula2014phase, prasaad2020handling}. The concept of analyzing transactions is extended in \ourSystem\ through static analysis of transaction types. However, in our system, this analysis is applied to transaction types rather than batches of predefined transactions, making the scope of the application distinct from previous approaches.

ORTHRUS~\cite{ren2016deadlockfree} is a database system that ensures deadlock avoidance through planned data access. ORTHRUS achieves deadlock-freedom by utilizing the lexicographical ordering of lock acquisition. In contrast, \ourSystem\ achieves deadlock-freedom through a dynamic ordering of lock acquisition tailored to the transactions. \ourSystem\ modifies the locking order only when necessary.

\four{
Aria~\cite{lu2020aria} is an optimistic, batch-oriented system that avoids requiring a priori knowledge of transaction read/write sets by performing its analysis online after transactions have executed; conflicts are resolved in a commit phase, and transactions may abort as needed.
In contrast, \ourSystem\ performs static analysis on transaction logic before execution, without requiring runtime information, enabling deadlock-free and abort-free 2PL.
Aria reorders commit order by transforming RAW to WAR dependencies to minimize aborts. Under high contention, Aria falls back to a Calvin-like protocol using the now-known read/write sets.  In contrast, \ourSystem\ reorders lock acquisitions in advance to break dependency cycles.
Other transaction scheduling techniques have been proposed to address the problem of contention, either by assigning a schedule based on arrival order~\cite{bernstein1983multiversion, cheng2024towards} or by reordering the schedule after transaction commit and/or abort~\cite{burke2023morty, dong2023fine}.
}

\cat{Early Write Visibility}
The authors in~\cite{faleiro2017earlyWriteVisibility} propose PWV, which breaks transactions into multiple atomic pieces and builds a dependency graph to schedule and commit transactions with finer granularity. Similar ideas are applied in various other works in the context of deterministic databases~\cite{prasaad2020handling, faleiro2014bohm, ding2018improving}, where transactions are executed in finer-grained pieces. 
These protocols rely on the assumption of determinism, where transaction execution is ordered prior to execution. In contrast, \ourSystem\ does not rely on the assumption of determinism.

Previous research has proposed enabling dirty reads of writes~\cite{reddy2004speculative, guo2021bamboo, gupta1997revisiting, jones2010low, sadoghi2014reducing}. Authors in~\cite{reddy2004speculative} propose an early write visibility discipline where a transaction executes with both the pre-image and after-image of reading uncommitted data. In~\cite{guo2021bamboo}, the authors propose Bamboo, which enables early reads of uncommitted writes by keeping track of dirty reads. In these approaches, the problem of cascading aborts persists, and additional runtime overhead is introduced to support early write visibility.

Transaction chopping techniques, such as those proposed in~\cite{shasha1995transactionChopping, zhang2013transactionChain, xie2014salt, wang2016ics3}, aim to decompose transactions into smaller subtransactions while preserving serializability. However, these methods often suffer from rigid conditions, making it challenging to achieve coarse-grained decomposition. Recent advancements like Lynx~\cite{zhang2013transactionChain} and Salt~\cite{xie2014salt} attempt to address this by leveraging application semantics to reduce conflicts among subtransactions. Despite these improvements, they lack a dynamic mechanism to relax transaction chopping conditions. IC3~\cite{wang2016ics3}, which introduces early write visibility for transactions through optimistic concurrency control (OCC), inherits the general limitations of transaction chopping. Additionally, IC3 and other OCC approaches~\cite{wang2016mocc, mahmoud2014maat} are not suitable for high-contention workloads due to the high abort rates inherent to them. \ourSystem\ introduces partial transaction chopping, which allows flexibility to support more workloads and targets high-contention workloads.

\cat{Hybrid Approaches}
Prior work has explored combining locking and OCC. For instance, Static and Dynamic OCC~\cite{yu1992analysis} switch from OCC to locking after the first retry. MOCC~\cite{wang2016mocc} and Hsync~\cite{shang2016hsync} are other examples that propose hybrid 2PL–OCC protocols. ACC~\cite{tang2017acc} goes further by adaptively selecting from a broader range of protocols. 
CCaaLF~\cite{pan2025ccaalf} learns an agent function, implemented as an efficient in-database lookup table, that maps database states to concurrency control actions.
However, these methods are generally designed for dynamic workloads and do not directly address hotspots, as they typically delay making writes visible after the transaction is committed.

\section{Conclusion}
\label{sec:Conclusion}
In this paper, we presented \ourSystem, a deadlock-free 2PL protocol designed to tolerate high-contention workloads. 
\ourSystem\ statically analyzes transactions using the \graph\ representation and applies lock manipulation techniques, such as moving locks earlier or combining locking nodes, to eliminate deadlocks. 
To further reduce contention, \ourSystem\ introduces a novel partial transaction chopping mechanism that enables the early release of locks, enhancing transaction concurrency.
Our evaluation results show that \ourSystem\ outperforms state-of-the-art concurrency control protocols, achieving a performance improvement of \re{\(2.86\times\)} and reducing tail latency by \re{\(48\%\)} in the TPC-C workload.
%

\vspace{-5pt}
\section{Acknowledgement}
This research is partly supported by a gift from Roblox and the NSF under grants 2321121, CNS1815212, and SaTC-2245372.
\appendix

\section{TPC-C Workload Analysis}
\label{appendix:tpcc}

\vspace{-5pt}
\cat{Static Analysis}
Figure~\ref{fig:tpcc_slw} illustrates the \graph\ representation for \textit{New Order} and \textit{Payment} transactions. 
After applying the read-write constraint to the \graph, potential deadlocks are eliminated by moving the Customer lock earlier in the \textit{Payment} transaction.

\begin{figure}[t]
    \centering
    \includegraphics[width=0.7\linewidth]{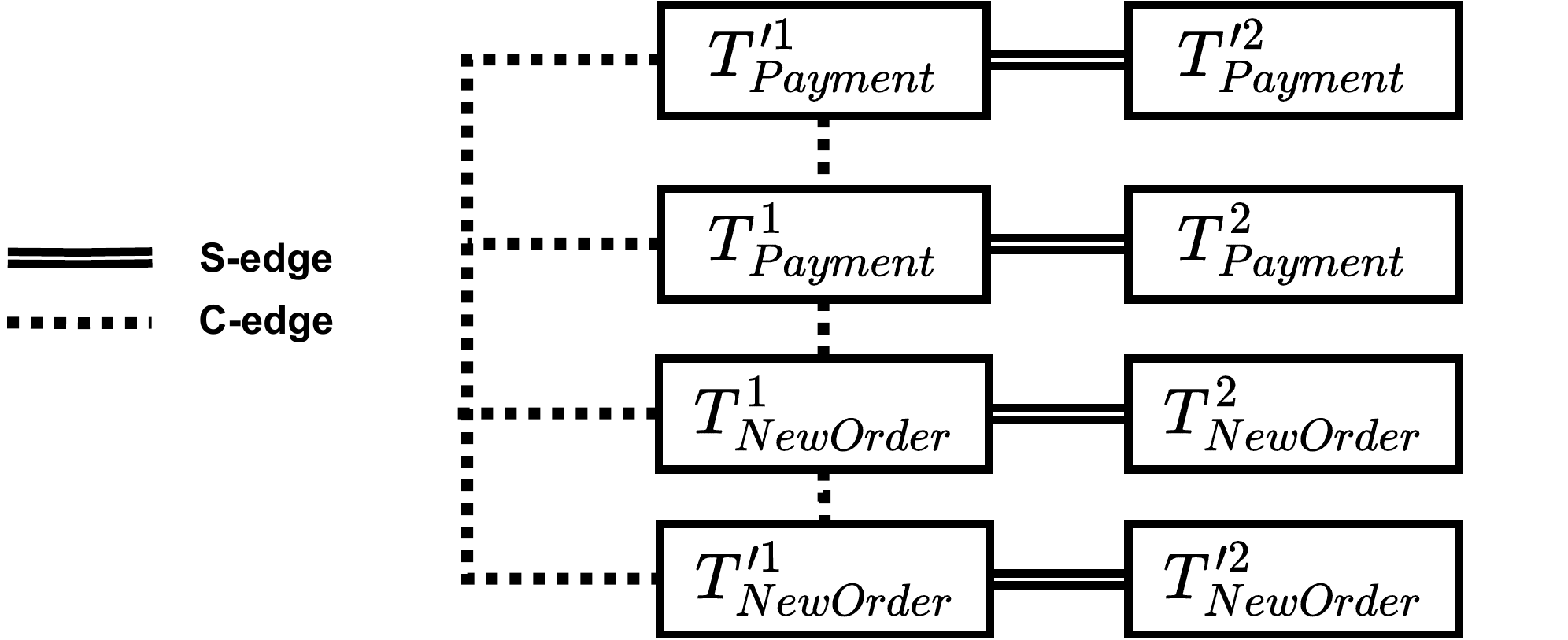}
    \caption{SC-Graph representation of the TPC-C workload after the early release of the \textit{District} and \textit{Warehouse} locks in the \textit{Payment} transaction and the \textit{Customer}, \textit{District}, and \textit{Warehouse} locks in the \textit{New Order} transaction}
    \vspace{-5pt}
    \label{fig:tpcc_sc}
\end{figure}

\cat{Partial Transaction Chopping}
\ourSystem's static analysis releases the locks on \textit{District} and \textit{Warehouse} after the last write on \textit{District} in the \textit{Payment} transaction, resulting in two subtransactions. 
Similarly, in the \textit{New Order} transaction, the locks on \textit{Customer}, \textit{District}, and \textit{Warehouse} are released after the last write on \textit{District}, by partially chopping the transaction into two subtransactions. To remove self-conflicts between the second subtransaction of the \textit{New Order} transaction, the \textit{Stock} lock is moved to the first subtransaction.
Our partial transaction chopping method, unlike traditional transaction chopping, enables the creation of an SC-acyclic SC-Graph for TPC-C transactions, thereby reducing contention.
Figure~\ref{fig:tpcc_sc} illustrates the SC-Graph for these choppings after applying early lock release and \textit{Stock} lock movement. This SC-Graph does not contain any SC-cycles, ensuring that the final execution is serializable.

\newpage
\bibliographystyle{ACM-Reference-Format}
\bibliography{sample-base}


\end{document}